\documentclass[12pt,hidelinks]{article}
\usepackage[sectionbib,sort]{natbib}
\usepackage{array,epsfig,fancyheadings,rotating}
\usepackage[]{hyperref}
\usepackage{sectsty, secdot}
\usepackage[margin=1in]{geometry}

\usepackage{amsmath}
\usepackage{amssymb}
\usepackage{amsfonts}
\usepackage{multirow}
\usepackage{amsthm}

\usepackage{bm, enumerate,xcolor}
\usepackage{graphicx,multicol}
\usepackage{float,subcaption}
\usepackage{multirow}
\usepackage{epstopdf}
\usepackage[plain,noend]{algorithm2e}

\newcommand{\tr}{\mathrm{tr}}
\newcommand{\onen}{\frac{1}{n}}
\newcommand{\sumn}{\sum_{i=1}^n}
\graphicspath{{sim_plot/}}

\newcommand{\purple}{\color{black}}
\usepackage{xr}
\allowdisplaybreaks

\newtheorem{theorem}{Theorem}
\newtheorem{lemma}{Lemma}

\theoremstyle{definition}

\newtheorem{remark}{Remark}
\begin{document}
\title{Optimal Subsampling Algorithms for Big Data Regressions}
\author{Mingyao Ai$^1$, Jun Yu$^{1,2}$, Huiming Zhang$^1$, HaiYing Wang$^3$\\[9mm]
  LMAM, School of Mathematical Sciences and Center\\
  for Statistical Science, Peking University $^1$\\[2mm]
  School of Mathematics and Statistics, Beijing Institute of Technology $^2$\\[2mm]
Department of Statistics University of Connecticut $^3$\\
  haiying.wang@uconn.edu}
\date{}
\maketitle
\begin{abstract}
To fast approximate maximum likelihood estimators with massive data, this paper studies the {\bf o}ptimal {\bf s}ubsampling {\bf m}ethod under the {\bf A}-optimality {\bf c}riterion (OSMAC) for generalized linear models. The consistency and asymptotic normality of the estimator from a general subsampling algorithm are established, and optimal subsampling probabilities under the A- and L-optimality criteria are derived. Furthermore, using Frobenius norm matrix concentration inequalities, finite sample properties of the subsample estimator based on optimal subsampling probabilities are also derived. Since the optimal subsampling probabilities depend on the full data estimate, an adaptive two-step algorithm is developed. Asymptotic normality and optimality of the estimator from this adaptive algorithm are established.
The proposed methods are illustrated and evaluated through numerical experiments on simulated and real datasets.
\end{abstract}

\noindent {\it Key words:} generalized linear models; massive data; matrix concentration inequality. 

\section{Introduction}\label{sec:intro}

Nowadays, massive data sets are ubiquitous in many scientific fields and practices such as in astronomy, economics, and industrial problems. Extracting useful information from these large data sets is a core challenge for different communities including computer science, machine learning, and statistics. Over last decades, progresses have been made through various investigations to meet this challenge. However, computational limitations still exist due to the faster growing pace of data volumes. For this, subsampling is a popular technique to extract useful information from massive data. This paper focuses on this technique and will develop optimal subsampling strategies for generalized linear models (GLMs).  Typically the maximum likelihood estimators (MLE) 
 are found numerically by using the Newton-Raphson method.
However, fitting a GLM on massive data is not an easy task through the iterative Newton-Raphson method, and it requires $O(p^2n)$ time in each iteration of the optimization procedure.

An efficient way to solve this problem  is  the subsampling method \citep[see][as an example]{drineas2006sampling} as this method essentially downsizes the data volume.
\cite{Drineas2011Faster} proposed to make a randomized Hadamard transform on data and then use uniform subsampling to take random subsamples to approximate ordinary least square estimators in linear regression models. \cite{Ma2015A}, \cite{Ma2015review} developed an effective subsampling method for linear regression models, which uses normalized statistical leverage scores of the covariate matrix as non-uniform subsampling probabilities. \cite{Jia2014Influence} studied leverage sampling for GLMs based on generalized statistical leverage scores. \cite{Wang2017Optimal} and \cite{yao2018optimal} developed an optimal subsampling procedure to minimize the asymptotic  mean squared error (MSE) of the resultant subsample-estimator given the full data which is based on $A$- or $L$-optimality criterion in the language of optimal design. \cite{Wang2017Information} proposed a new algorithm called information-based optimal subdata selection method for linear regressions on big data.
The basic idea is to select the most informative data points deterministically based on $D$-optimality without relaying on random subsampling. A divide-and-conquer version of the algorithm was developed in \cite{wang2019divide}. Recent developments of big data subsampling method can be found in
\cite{Wang2016review}.

Methodological investigations on subsampling methods with statistical guarantees for massive data regression are still limited when models are complex. {To the best of our knowledge, most of the existing results concern linear regression models such as in \cite{Ma2015A} and \cite{Wang2017Information}. The optimal subsampling method in \cite{Wang2017Optimal} and \cite{yao2018optimal} is designed specifically for logistic and multinomial regression models, respectively.} { However, only linear and logistic regressions are not enough to meet practical needs \citep{Claudia2000Noncanonical}. For example, we may need Poisson  or negative binomial distribution for count data and need Gamma or inverse Gaussian distribution for data with non-negative responses.} In addition, the aforementioned investigations did not consider finite sample properties of subsampled estimators. 
In this paper, we fill these gaps by deriving optimal subsampling probabilities for GLMs, including these with non-canonical link functions { which allow for a wide range of statistical models for regression analysis.} Furthermore, we will derive finite-sample upper bounds for approximation errors that can be practically used to make the balance between the subsample size and prediction accuracy. Due to the non-natural link, our investigation is substantially distinct from the that in \cite{Wang2017Optimal}. For example, the Hessian matrix in the models considered in this paper may be dependent on responses.  

The rest of this paper is organized as follows. Section \ref{sec:model} introduces the model setup and derives asymptotic properties for the general subsampling estimator.
 Section \ref{sec:appr-optim-subs} derives  optimal subsampling strategies based on $A$- and $L$-optimality criteria  for GLMs. Finite-sample error bounds are also derived in this section.
Section \ref{sec:two-step} designs a two-step algorithm to approximate the optimal subsampling procedure and obtains asymptotic properties of the resultant estimator. Section \ref{sec:experiments} illustrates our methodology through numerical simulations and a real data applications.  

\section{Preliminaries}\label{sec:model}

\subsection{Models and Assumptions}\label{sec:glm}
Recall the definition of one parameter exponential family of distributions $f(y|\theta)=h(y)\exp(\theta y -\psi(\theta)), \theta \in \Theta $
as in (5.50) of \cite{Efron2016}, where $\theta$ is called the canonical parameter and $\Theta$ is called the natural parameter space. Here $f(\cdot|\theta)$ is a probability density function for the continuous case or a probability mass function for the discrete case; $h(\cdot)$ is a specific function that does not depend on $\theta$; and the parameter space $\Theta$ is defined as $\Theta:= \{ \theta \in \mathbb{R}:\: \int h(x) \exp(\theta x)\mu(dx)<\infty \}$ with $\mu$ being the dominating measure. 
 The exponential family includes most of the commonly used distributions such as normal, gamma, Poisson, and binomial distributions \citep[see][]{Mccullagh1989Generalized,Efron2016}.

A key tactic for a generalized linear regression model is to express $\theta$ in form of a linear function of regression coefficients.
Let $(\bm x,y)$ be a pair of random variables where $y \in\mathbb{R}$ and $\bm x\in \mathbb{R}^{p}$. The generalized linear regression model assumes that the conditional distribution of $y_i$ given $\bm x_i$ is determined by $\theta_i=u({{\bm\beta}^{T}{\bm x}_{i}})$. Specifically for exponential family, it assumes that the distribution of $y\vert \bm x$ is
\begin{equation}\label{eq:1}
  f(y\vert\bm\beta,\bm x)=h(y)\exp(yu({{\bm\beta}^{T}{\bm x}_{i}})-\psi(u({{\bm\beta}^{T}{\bm x}_{i}}))),
  \quad\text{ with }\quad {\bm\beta}^{T}\bm x \in \Theta.
\end{equation}
The problem of interest is to estimate the unknown $\bm\beta$ from observed data.
As special case when $u(t)=t$, the corresponding models are the so-called GLMs with canonical link functions. Some typical examples of this type GLMs are logistic regression for binary data and Poisson regression for count data. {A commonly used GLM with non-canonical link function is negative binomial regression (NBR), which is often used as an alternative to Poisson regression when data exhibit overdispersion. 
For this model, $u(t) = t - \log (\nu  + e^t)$ and $\psi (u(t)) = \nu \log (\nu + e^t)$ for some size parameter $\nu$.}

\subsection{General Subsampling Algorithm and its Asymptotic Properties}\label{sec:subsamplingalg}
In this subsection, we study a general subsampling algorithm for GLMs and obtain some asymptotic results.

{To facilitate the presentation, denote the full data matrix by $\mathcal{F}_n=(\bm X, \bm y)$, where
$\bm X=(\bm x_1, \dots,\bm x_n)^T$ 
is the covariate matrix and $\bm y=(y_1,\dots,y_n)^T$ is the response vector. In this paper, we assume that $(\bm x_i, y_i)$'s are independently generated from a GLM.}
Let $S$ be a set of subsample with $r$ data points, and define the sampling distribution $\pi_i$ for all data points $i=1,2,...n$ as $\bm \pi$. 
A general subsampling algorithm follows the steps below.
\begin{enumerate}
  \item Assign a sampling distribution $\bm \pi$ such that in each draw, the $i$-th element in the full dataset $\mathcal{F}_n$ has the inclusion probability $\pi_i$.
  \item Sample with replacement $r$ times to form the subsample set $S := \{(y_{i}^*,\bm {x}_{i}^*,\pi_{i}^*),i=1,\ldots,r\}$, where $\bm x^{*}_i$, $y^{*}_i$, and $\pi_i^{*}$ stand for covariates, responses, and subsampling probabilities in the subsample, respectively.
  \item Based on the subsample set $S$, calculate the weighted log-likelihood estimator by maximizing the following function
  \begin{equation}\label{eq:reweigt2}
      L^*(\bm\beta)=\frac{1}{r}\sum^r_{t=1}\frac{1}{{\pi}_{i}^*}[y_{i}^*u({\bm \beta}^{T}\bm x_{i}^*)-\psi(u({\bm \beta}^{T}\bm x_{i}^*))].
    \end{equation}
\end{enumerate}

An important feature of the above algorithm  is that subsample estimator is essentially a weighted MLE and the corresponding weights are inverses of subsampling probabilities.
 This is analogous to the Hansen-Hurwitz estimator \citep{Hansen1943On} in classic sampling techniques.
 For an overview see \cite{Sarndal1992Model}.
 Although \cite{Ma2015A} showed that the unweighted subsample estimator is asymptotically unbiased for $\bm\beta$ in leveraging sampling, an unweighted subsample estimator is in general biased if the sampling distribution $\bm \pi$ depends on the responses. The inverse-probability weighting scheme is to remove bias, and we restrict our analysis on the weighted estimator here.

Let $\dot\psi(t)$ and $\ddot\psi(t)$ be the first and the second derivatives of $\psi(t)$, respectively. To characterize asymptotic properties of subsampled estimators, we require some regularity assumptions listed below. 
 \begin{itemize}
\item [] (H.1):  Assume that ${\bm\beta}^{T}\bm x$ lies in the interior of a compact set $K \in \Theta$ almost surely.

\item [] (H.2): The regression coefficient $\bm\beta$ is a inner point of the compact domain $\Lambda _B=\left\lbrace  \bm\beta \in \mathbb{R}^{p}: \|\bm\beta\| \le B \right\rbrace$ for some constant $B$.

\item [] (H.3): Central moments condition:  $n^{-1}\sum^n_{i=1}|y_i-\dot{\psi}(u(\bm\beta^T\bm x_i))|^4 ={O_P}(1)$
for all $\bm\beta \in \Lambda _B$.

\item [] (H.4): As $n\rightarrow\infty$, the observed information matrix \[\begin{array}{l}
    \mathcal{J}_X:=\frac{1}{n}\sum\limits_{i = 1}^n {\{ \ddot u(} {\hat{\bm\beta}_{\rm MLE} ^T}{\bm x_i}){\bm x_i}\bm x_i^T[\dot \psi (u({\hat{\bm\beta}_{\rm MLE} ^T}{\bm x_i})) - {y_i}] \\
    \hspace{5cm}+ \ddot \psi (u({\hat{\bm\beta}_{\rm MLE} ^T}{\bm x_i})){\dot u^2}({\hat{\bm\beta}_{\rm MLE} ^T}{\bm x_i}){\bm x_i}\bm x_i^T]\}
\end{array}\]
goes to a positive-definite matrix in probability.

\item [] (H.5): Require that the full sample { covariates} have finite 6th-order moments, i.e., $E{{{\left\| {{\bm x_1}} \right\|}^6}} {\rm{ \le }}\infty.$

\item [] (H.6): Assume
${n^{-2}}\sum^n_{i=1}
     \|\bm x_i\|^k/\pi_i={O_P}(1)$ for $k=2,4$.

\item [] (H.7): For $\gamma=0$ and some $\gamma>0$, assume
$$\frac{1}{n^{2+\gamma}}
    \sum^n_{i=1}\frac{|y_i-\dot{\psi}_i(u(\hat{\bm\beta}_{\rm MLE}^T\bm x_i))|^{2+\gamma}
    \|\dot{u}(\hat{\bm\beta}_{\rm MLE}^T\bm x_i)\bm x_i\|^{2+\gamma}}{\pi_i^{1+\gamma}}={O_P}(1)
.$$

\end{itemize}

Here, assumptions (H.1) and (H.2) are the set of assumptions used in \cite{clemencon2014}.  
The set in (H.2) is also called admissible set which 
premises for consistency estimation for GLMs with full data \citep[see][]{Fahrmeir1985}. 
{These two assumptions ensure that $E(y_i|\bm x_i)<\infty$ for all $i$.}
{Assumption (H.4) imposes a condition on the covariates to make sure that the MLE based on the full dataset is consistent. To obtain the Bahadur representation of the subsampled estimator, (H.3) and (H.5) are needed.} {  Assumptions (H.6) and (H.7) are moment conditions on covariates and sub-sampling probabilities. 
Assumption (H.7) is required by the Lindeberg-Feller central limit theorem.
Specifically for the uniform subsampling with $\pi_i=n^{-1}$ or more generally when $\max_{i=1,\ldots,n}(n\pi_i)^{-1}=O_P(1)$, (H.7) is implied by that $n^{-1}
    \sum^n_{i=1}{|y_i-\dot{\psi}_i(u(\hat{\bm\beta}_{\rm MLE}^T\bm x_i))|^{2+\gamma}
    \|\dot{u}(\hat{\bm\beta}_{\rm MLE}^T\bm x_i)\bm x_i\|^{2+\gamma}}={O_P}(1),$ which is guaranteed by the conditions that $E|y|^{4+2\gamma}=O(1)$ 
under (H.1) and (H.5).} 

 The theorem below presents the consistency of the estimator from the subsampling algorithm to the full data MLE.
 \begin{theorem}\label{thm:as-general-alg}
  If Assumptions (H.1)--(H.6) hold, then as $n\rightarrow\infty$ and $r\rightarrow\infty$, $\tilde{\bm\beta}$ is consistent to $\hat{\bm\beta}_{\rm MLE}$ in conditional probability
  given $\mathcal{F}_n$. Moreover, the rate of convergence is $r^{-1/2}$.
  That is, with probability approaching one, for any $\epsilon>0$, there exist  finite $\Delta_\epsilon$ and $r_\epsilon$ such that
  \begin{equation}\label{eq:18}
    P(\|\tilde{\bm\beta}-\hat{\bm\beta}_{\rm MLE}\|\ge
    r^{-1/2}\Delta_\epsilon|\mathcal{F}_n)<\epsilon
  \end{equation}
  for all $r>r_\epsilon$.
\end{theorem}

Besides consistency, we derive the asymptotic distribution of the approximation error, and prove that the approximation error,
$\tilde{\bm\beta}-\hat{\bm\beta}_{\rm MLE}$, is asymptotically normal in conditional distribution.

\begin{theorem}\label{thm:CLT}
  If Assumptions (H.1)--(H.7) hold, then as $n\rightarrow\infty$ and $r\rightarrow\infty$, conditional on $\mathcal{F}_n$ in probability,
  \begin{equation}\label{normal}
    V^{-1/2}(\tilde{\bm\beta}-\hat{\bm\beta}_{\rm MLE})
    \longrightarrow N(0,I)
  \end{equation}
  in distribution, where $V=\mathcal{J}_X^{-1}V_c\mathcal{J}_X^{-1}=O_{p}(r^{-1})$ and
  \begin{equation}\label{varorder}
    V_c=\frac{1}{rn^2}
    \sum_{i=1}^n\frac{\{y_i-\dot{\psi}(u(\hat{\bm\beta}_{\rm MLE}^T\bm x_i))\}^2{\dot u^2}(\hat{\bm\beta}_{\rm MLE}^T\bm x_i)\bm x_i\bm x_i^T}{\pi_i}.
  \end{equation}
\end{theorem}

\section{Optimal Subsampling Strategies}\label{sec:appr-optim-subs}

In this section, we will consider how to specify subsampling distribution $\bm\pi=\{\pi_i\}_{i=1}^{n}$ with theoretical backup. 

\subsection{Optimal Subsampling Strategies Based on Optimal Design Criteria}\label{sec:mV}
Based on $A$-optimality criterion in the theory of design of experiments \citep[see][]{pukelsheim2006optimal}, optimal subsampling  is to choose subsampling probabilities such that the asymptotic MSE of $\tilde{\bm\beta}$ is minimized. This idea was proposed  in \cite{Wang2017Optimal} and we call the resulting subsampling strategy {\it mV-optimal}.

\begin{theorem}\label{thm:3}
  { The} subsampling strategy is {\it mV-optimal} if the subsampling probability is chosen such that
  \begin{equation}\label{eq:pi-amse-w}
    \pi_i^{\mathrm{mV}}=
    \frac{|y_i-\dot{\psi}(u(\hat{\bm\beta}_{\rm MLE}^T\bm x_i))|\|\mathcal{J}_X^{-1}{\dot u}(\hat{\bm\beta}_{\rm MLE}^T\bm x_i)\bm x_i\|}
    {\sum_{j=1}^n|y_j-\dot{\psi}(u(\hat{\bm\beta}_{\rm MLE}^T\bm x_i))|\|\mathcal{J}_X^{-1}{\dot u}(\hat{\bm\beta}_{\rm MLE}^T\bm x_j)\bm x_j\|},\;
    i=1,2,...,n.
  \end{equation}
\end{theorem}

  The optimal subsampling probability $\bm{\pi}^{\mathrm{mV}}$  has a meaningful interpretation from the view-point of optimal design of experiments \citep{pukelsheim2006optimal}.

  Note that under mild condition the  ``empirical information matrix'' ${\mathcal{J}^e_X}=\frac{1}{n}\sum_{i=1}^n[y_i-\dot{\psi}(u(\hat{\bm\beta}_{\rm MLE}^T\bm x_i))]^2{\dot u^2}(\hat{\bm\beta}_{\rm MLE}^T\bm x_i)\bm x_i\bm x_i^T$ and $\mathcal{J}_X$ converge to the same limit, the Fisher information matrix of model~\eqref{eq:1}. This means that ${\mathcal{J}^e_X}-\mathcal{J}_X=o_P(1)$. Thus, $\mathcal{J}_X$ can be replaced by ${\mathcal{J}^e_X}$ in $\bm{\pi}^{\mathrm{mV}}$, because Theorem~\ref{thm:CLT} still holds if $\mathcal{J}_X$ is replaced by ${\mathcal{J}^e_X}$ in \eqref{varorder}.
    Let $\eta_{\bm x_i}=[y_i-\dot{\psi}(u(\hat{\bm\beta}_{\rm MLE}^T\bm x_i))]^2{\dot u^2}(\hat{\bm\beta}_{\rm MLE}^T\bm x_i)\bm x_i\bm x_i^T$, the contribution of the $i$-th observation to the empirical information matrix, and ${\mathcal{J}^e_{X\bm x_i\alpha}}=(1-\alpha){\mathcal{J}^e_X}+\alpha\eta_{\bm x_i}$, which can be interpreted as a movement of the information matrix in a direction determined by the $i$-th observation. The directional derivative of $\text{tr}({\mathcal{J}^e_X}^{-1})$ through the direction determined by the $i$th observation is $F_i=\lim_{\alpha\rightarrow0+} \alpha^{-1}\{\text{tr}({\mathcal{J}^e_X}^{-1})-\text{tr}({\mathcal{J}^e_{X\bm x_i\alpha}}^{-1})\}$. This directional derivative is used to measure the relative gain in estimation efficiency under the $A$-optimality by adding the $i$-th observations into the sample. 
Thus, the optimal subsampling strategy prefers to select the data points with large values of directional derivatives, i.e., data points that will result in a larger gain under the $A$-optimality. 

The optimal subsampling strategy  derived from $mV$-optimal critria requires the calculation of $\|\mathcal{J}_X^{-1}{\dot u}(\hat{\bm\beta}_{\rm MLE}^T\bm x_i)\bm x_i\|$ for $i=1,2,...,n$, which takes $O(np^2)$ time.
 To reduce the calculating time, \cite{Wang2017Optimal}  proposed a modified optimality criterion to minimize $\tr(V_c)$.
 This criterion essentially is the L-optimality criterion in optimal experimental design \citep[see][]{pukelsheim2006optimal}, which is to improve the quality of $\mathcal{J}_X\tilde{\bm\beta}$.
 It is easy to see that  only $O(np)$ time is needed to calculate the optimal sampling probabilities. We call the resulting subsampling strategy {\it mVc-optimal}.

\begin{theorem}\label{thm:5}
  { The} subsampling strategy is {\it mVc-optimal} if the subsampling probability is chosen such that
  \begin{equation}\label{eq:optimalBpi-w}
    \pi_i^{\mathrm{mVc}}=\frac{|y_i-\dot{\psi}(u(\hat{\bm\beta}_{\rm MLE}^T\bm x_i))|\|{\dot u}(\hat{\bm\beta}_{\rm MLE}^T\bm x_i)\bm x_i\|}
    {\sum_{j=1}^n|y_j-\dot{\psi}(u(\hat{\bm\beta}_{\rm MLE}^T\bm x_j))|\|{\dot u}(\hat{\bm\beta}_{\rm MLE}^T\bm x_j)\bm x_j\|},\ i=1,2, ..., n.
  \end{equation}
\end{theorem}

Note that in order to calculate  $\|\mathcal{J}_X^{-1}{\dot u}(\hat{\bm\beta}_{\rm MLE}^T\bm x_i)\bm x_i\|$ for $i=1,2,...,n$, we need $O(np^2)$ time while  only $O(np)$ time is needed to evaluate $\|{\dot u}(\hat{\bm\beta}_{\rm MLE}^T\bm x_i)\bm x_i\|$.
Note that $\mathcal{J}_X$ and $V_c$ are nonnegative definite, and $V=\mathcal{J}_X^{-1}V_c\mathcal{J}_X^{-1}$. Simple matrix algebra yields that ${\rm{tr}}(V) = {\rm{tr}}({V_c}{\cal J}_X^{ - 2}) \le {\sigma _{\max }}({\cal J}_X^{ - 2}){\rm{tr}}({V_c})$, where $\sigma_{\max}(A)$ denotes the maximum singular value of matrix $A$. 
Since $\sigma _{\max }({\cal J}_X^{ - 2})$ does not depend on $\bm{\pi}$, minimizing $\tr(V_c)$ minimizes an upper bound of $\tr(V)$.
 In fact, for two given subsampling strategies  $\bm{\pi}^{(1)}$ and $\bm{\pi}^{(2)}$, if $V_c(\bm{\pi}^{(1)})\le V_c(\bm{\pi}^{(2)})$ in the sense of Loewner-ordering, then it follows that $V(\bm{\pi}^{(1)})\le V(\bm{\pi}^{(2)})$.
Thus the alternative optimality criterion greatly reduces the computing time 
without losing too much estimation accuracy. 

Due to the score function for the log-likelihood,
it is interesting that $\pi_i^{\mathrm{mVc}}$'s in Theorem \ref{thm:5} are proportional to $\|\{y_i - \dot\psi(\hat{\bm\beta}_{\rm MLE}^T\bm x_i)\}{\dot u}(\hat{\bm\beta}_{\rm MLE}^T\bm x_i)\bm x_i\|$, norms of gradients of the log-likelihood at individual data points evaluated at the full data MLE. This is trying to find the subsample that best approximate the full data score function at the full data MLE.

We now illustrate Theorem~\ref{thm:3} and Theorem~\ref{thm:5} with some commonly used GLMs. 
Note that $u(\cdot)$ is the identity function for GLMs with nature link functions such as logistic and Poisson regressions.
For logistic regression,
\[\pi _i^{{\rm{mV}}} = \frac{|y_i - p_i|  \left\| {{{\mathcal{J}_X^{ - 1}}}{\bm x_i}} \right\|}{{\sum_{j = 1}^n {\left| {{y_j} - p_j} \right|  \left\| {{{\mathcal{J}_X^{ - 1}}}{\bm x_j}} \right\|} }},\qquad
  \pi _i^{{\rm{mVc}}} = \frac{{\left| {{y_i} - p_i} \right|  \left\| {{\bm x_i}} \right\|}}{{\sum_{j = 1}^n {\left| {{y_j} - p_j} \right|  \left\| {{\bm x_j}} \right\|} }},\]
with $p_i={{{\exp(\bm{\hat\beta}_{\rm MLE}^T{\bm x_i})}}}/\{{1 + {\exp(\bm{\hat\beta}_{\rm MLE}^T{\bm x_i})}}\}$ and $\mathcal{J}_X=n^{-1}\sum_{k=1}^n p_k(1-p_k){\bm x_k}\bm x_k^T$. These are the same as the results in \cite{Wang2017Optimal}. For Poisson regression,
\[\pi _i^{{\rm{mV}}} = \frac{{\left| {{y_i} - {\lambda_i}} \right|  \left\| {{{\mathcal{J}_X^{ - 1}}}{\bm x_i}} \right\|}}{{\sum_{j = 1}^n {\left| {{y_j} - {\lambda_j}} \right|  \left\| {{{\mathcal{J}_X^{ - 1}}}{\bm x_j}} \right\|} }},\qquad
  \pi _i^{{\rm{mVc}}} = \frac{{\left| {{y_i} - {\lambda_i}} \right|  \left\| {{\bm x_i}} \right\|}}{{\sum_{j = 1}^n {\left| {{y_j} - {\lambda_j}} \right|  \left\| {{\bm x_j}} \right\|} }},\]
with $\lambda_i=\exp({\bm{\hat\beta}_{\rm MLE}^T{\bm x_i}})$ and $\mathcal{J}_X=n^{-1}\sum_{k=1}^n\exp({\bm{\hat\beta}_{\rm MLE}^T{\bm x_k}}){\bm x_k}\bm x_k^T.$
NBR does not have a canonical link function, and the conditional distribution of the response is modeled by a two-parameter distribution
\begin{eqnarray}\nonumber
f(y_{i}|\nu,\mu_{i})=\frac{\Gamma(\nu+y_{i})}{\Gamma(\nu)y_{i}!}\left(\frac{\mu_{i}}{\nu+\mu_{i}}\right)^{y_{i}}\left(\frac{\nu}{\nu+\mu_{i}}\right)^{\nu}, \quad i=1,2,\ldots,n,
\end{eqnarray}
where the size parameter $\nu$ can be estimated as a nuisance parameter. The optimal subsampling probabilities for NBR with size parameter $\nu$ are
\[\begin{array}{l}
    \pi _i^{{\rm{mV}}} = \frac{{\left| {{y_i} - \mu_i} \right|  \left\| {{{\mathcal{J}_X^{ - 1}}}\frac{{\nu {\bm x_i}}}{{\nu  + \mu_i}}} \right\|}}{{\sum_{j = 1}^n {\left| {{y_j} - \mu_j} \right|  \left\| {{\mathcal{J}_X^{-1}}\frac{{\nu {\bm x_j}}}{{\nu  + \mu_j}}} \right\|} }},\qquad
    \pi _i^{{\rm{mVc}}} = \frac{{\left| {{y_i} - \mu_i} \right|  \left\| \frac{\nu {\bm x_i}}{\nu  + \mu_i} \right\|}}{{\sum_{j = 1}^n {\left| {{y_j} - \mu_j} \right|  \left\| \frac{\nu {\bm x_i}}{\nu  + \mu_j} \right\|} }}
\end{array},\]
with $\mu_i=\exp({\bm{\hat\beta}_{\rm MLE}^T{\bm x_i}})$ and $\mathcal{J}_X=n^{-1}\sum_{k = 1}^n \{\nu (\nu  + {y_i})\mu_i\}/(\nu  + \mu_i)^2{\bm x_k}\bm x_k^T. $

\subsection{Non-asymptotic Properties} 
\label{sec:r-size}

We derive some finite sample properties of the subsample estimators based on optimal subsampling probabilities $\bm{\pi}^{\mathrm{mV}}$ and $\bm{\pi}^{\mathrm{mVc}}$ in this section. Results are presented in forms of excess risks for approximating the mean responses and they hold for fixed $r$ and $n$ without requiring any quantity to go to infinity. These results show factors that affect the approximation accuracy.

Since $\dot{\psi}(u(\bm x_i^T\bm\beta))$ is the conditional expectation of the response $y_i$ given $\bm x_i$, we aim to characterize the quantity of $\tilde{\bm\beta}$ in prediction by examining $\|\dot{\psi}(u(\bm X_d^T\hat{\bm\beta}_{\rm MLE}))-\dot{\psi}(u(\bm X_d^T\tilde{\bm\beta}))\|$. This quantity is the distance between the estimated conditional mean responses based on the full data and that based on the subsamples. Intuitively, it measures the goodness of fit in using subsample estimator to predict the mean responses.
Note that we can always improve the accuracy of the estimator by increasing the subsample size $r$. Here we want to have a closer look at the effects of different quantities such as covariate matrix and data dimension and the effect of subsample size $r$ on approximation accuracy.

 Let $\sigma_{\max}(A)$ and $\sigma_{\min}(A)$ be the maximum and minimum non-zero singular values of matrix $A$, respectively,  $\kappa(A):=\sigma_{\max}(A)/\sigma_{\min}(A)$. { Denote} $\dot{\psi}(u(\bm X^T{\bm \beta}))$, the vector whose the $i$-th element is $\dot{\psi}(u(\bm x_i^T{\bm \beta}))$ and define $\dot u(\bm X^T{\bm \beta}): = {\rm{diag}}\{ \dot u(\bm x_1^T{\bm \beta}), \cdots ,\dot u(\bm x_n^T{\bm \beta})\} $. 
  For the estimator $\tilde{\bm \beta}$ obtained from the algorithm in Section 2 based on the subsampling probabilities, $\bm{\pi}^{\mathrm{mV}}$ and $\bm{\pi}^{\mathrm{mVc}}$,
  the following theorem holds.

 \begin{theorem}\label{thm:r-size}
 Let $\bm{\tilde X}$ denotes the design matrix consisting of  subsample covariates with each sampled element rescaled by $1/\sqrt{r\pi_i^*}$.  Assume that $\sigma_{\min}^2(\dot{u}(\bm X^T\tilde{\bm \beta})\bm{\tilde{ X}})\ge 0.5\sigma^2_{\min}(\dot{u}(\bm X^T\tilde{\bm \beta})\bm X)$, and
 both $\sigma_{\max}(\dot{u}(\bm X^T\tilde{\bm \beta})\bm X)/\sqrt{n}$ and $\sigma_{\min}(\dot{u}(\bm X^T\tilde{\bm \beta})\bm X)/\sqrt{n}$ are bounded.
 For any given $\epsilon\in(0,1/3)$,  with probability at least $1-\epsilon$, we have
 \begin{align}\label{eq:r-size}
 &~~~~\|\dot{\psi}(u(\bm X^T\hat{\bm\beta}_{\rm MLE}))-\dot{\psi}(u(\bm X^T\tilde{\bm\beta}))\|\nonumber\\
 & \le 2C_{\dot{u}}{[1+\frac{4\alpha\sqrt{\log(1/\epsilon)}}{\sqrt{r}}]{\sqrt p \kappa^2(\dot{u}(\bm X^T\tilde{\bm \beta})\bm X)}}\|[\bm y-\dot{\psi}(u(\bm X^T\hat{\bm\beta}_{\rm MLE}))]\|.
 \end{align}
 where 
 $\alpha=\kappa(\mathcal{J}_X^{-1})$ for $\bm{\pi}^{\mathrm{mV}}$ and $\alpha=1$ for $\bm{\pi}^{\mathrm{mVc}}$ and ${C_{\dot u}} = \mathop {\sup }\limits_{r \in K \subset \Theta } \left| {\dot u(r)} \right|$.
 \end{theorem}

 Theorem \ref{thm:r-size} not only indicates that the accuracy increases with subsample size $r$, which agrees with the results in Theorem \ref{thm:as-general-alg}, but also enables us to have a closer look at the effects of different quantities such as covariate matrix and data dimension and the effect of subsample size $r$ on approximation accuracy.
 Heuristically, the condition number of $\dot{u}(\bm X^T\tilde{\bm \beta})\bm X$ measures the collinearity of covariates in the full data covariate matrix; $p$ shows the curse of dimensionality; and $\|\bm y-\dot{\psi}(u(\bm X^T\hat{\bm\beta}_{\rm MLE}))\|$ measures the goodness of fit of the underlying model on the full data.

The result in (\ref{eq:r-size}) also indicates that we should choose $r\propto p$ to control the error bound, hence it seems reasonable to choose the subsample size as $r=cp$. This agrees with the
recommendation of choosing a sample size as large as 10 times of the number of covariates in \cite{Chapman1994Arctic} 
and \cite{Loeppky2009Choosing} for designed experiments.  
However, for designed experiments, covariate matrices are often orthogonal or close to be orthogonal, so $\kappa(\dot{u}(\bm X^T\tilde{\bm \beta})\bm X)$ is equal or close to 1 in these cases. For the current paper, full data may not be obtained from well designed experiments so $\dot{u}(\bm X^T\tilde{\bm \beta})\bm X$ may vary a lot. Thus, $\kappa(\dot{u}(\bm X^T\tilde{\bm \beta})\bm X)$ should also be considered in determining required subsample size for a given level of prediction accuracy.  

 The particular constant $0.5$ in Theorem \ref{thm:r-size}'s condition $\sigma_{\min}^2(\dot{u}(\bm X^T\tilde{\bm \beta})\bm{\tilde{ X}})\ge 0.5\sigma^2_{\min}(\dot{u}(\bm X^T\tilde{\bm \beta})\bm X)$  { can be replaced by any constant between 0 and 1. Here we follow the setting of \cite{Drineas2011Faster} and choose 0.5 for convenience.}
 This condition indicates that the rank of $\dot{u}(\bm X^T\tilde{\bm \beta})\bm{\tilde{X}}$ is the same as that of $\dot{u}(\bm X^T\tilde{\bm \beta})\bm X$. More details and interpretations about this condition can be found in \cite{Mahoney2012Randomized}.

 Using similar argument as in the proof of Theorem \ref{thm:r-size},
 it is proved that this condition holds with high probability.

 \begin{theorem}\label{thm:r-size1}
 Let $\dot{u}(\bm X^T\tilde{\bm \beta})\bm{\tilde X}$ denote the design matrix consisting of subsamples with each sampled element rescaled by $1/\sqrt{r\pi_i^*}$.
 Assume that {$|y_i-\dot{\psi}(u(\hat{\bm\beta}_{\rm MLE}^T\bm x_i))|\|{\dot u}(\hat{\bm\beta}_{\rm MLE}^T\bm x_i)\bm x_i\|\ge \gamma\|\bm x_i\|$} for all $i$ and $\sigma_{\max}(\dot{u}(\bm X^T\tilde{\bm \beta})\bm X)/\sqrt{n}$, $\sigma_{\min}(\dot{u}(\bm X^T{\bm \beta})\bm X)/\sqrt{n}$ are bounded. For any given $\epsilon\in(0,1/3)$,
 let  ${c_d} \le 1$ be a constant depending on $\dot{u}(\bm X^T\tilde{\bm \beta})\bm{X}$,   ${C_{\dot u}} = \mathop {\sup }\limits_{r \in K \subset \Theta } \left| {\dot u(r)} \right|$ and $r>64{c_d^2C_{\dot{u}}^2}\log(1/\epsilon)\sigma_{\max}^4(\bm X)p^2/(\alpha^2\delta^2\sigma_{\min}^4(\dot{u}(\bm X^T\tilde{\bm \beta})\bm X))$ { where $\delta$ is some constant depending on $\gamma$ and $\|[\bm y-\dot{\psi}(u(\bm X^T\hat{\bm\beta}_{\rm MLE}))]\|$}.
 Then  with probability at least $1-\epsilon$:
 \[\sigma_{\min}^2(\dot{u}(\bm X^T\tilde{\bm \beta})\bm{\tilde{ X}})\ge 0.5\sigma^2_{\min}(\dot{u}(\bm X^T\tilde{\bm \beta})\bm X),\]
 where $\alpha=\kappa(\mathcal{J}_X^{-1})$ for $\bm{\pi}^{\mathrm{mV}}$ and $\alpha=1$ for $\bm{\pi}^{\mathrm{mVc}}$.
 \end{theorem}

 \section{Practical Consideration and Implementation}\label{sec:two-step}

 For practical implementation, the optimal subsampling probabilities $\pi_i^{\mathrm{mV}}$'s and $\pi_i^{\mathrm{mVc}}$'s cannot be used directly because they depend on the unknown full data MLE, $\hat{\bm\beta}_{\rm MLE}$. 
 As suggested in \cite{Wang2017Optimal}, in order to calculate $\bm \pi^{\mathrm{mV}}$ or $\bm\pi^{\mathrm{mVc}}$, a pilot estimator of $\hat{\bm\beta}_{\rm MLE}$ has to be used. Let $\tilde{\bm\beta}_0$ be a pilot estimator based on a subsample of size $r_0$. It can be used to replace $\hat{\bm\beta}_{\rm MLE}$ in $\bm\pi^{\mathrm{mV}}$ or $\bm\pi^{\mathrm{mVc}}$, which then can be used to sample more informative subsamples.

 From the expression of $\bm\pi^{\mathrm{mV}}$ or $\bm\pi^{\mathrm{mVc}}$, the approximated optimal subsampling probabilities are both proportional to $|y_i-\dot{\psi}(u(\tilde{\bm\beta}_0^T\bm x_i))|$, so data points with $y_i\approx\dot{\psi}(u(\tilde{\bm\beta}_0^T\bm x_i))$ have very small probabilities to be selected and data points with  $y_i=\dot{\psi}(u(\tilde{\bm\beta}_0^T\bm x_i))$ will never be included in a subsample.
 On the other hand, if these data points are included in the subsample, the weighted log-likelihood function in (\ref{eq:reweigt2}) may be dominated by them. As a result, the subsample estimator may be sensitive to these data points. \cite{Ma2015A} also noticed the problem that some extremely small subsampling probabilities may inflate the variance of the subsampling estimator in the context of leveraging sampling.

 To protect the weighted log-likelihood function  from being inflated by these data points in practical implementation, we propose to set a threshold, say $\delta$, for $|y_i-\dot{\psi}(u(\tilde{\bm\beta}_0^T\bm x_i))|$, i.e., use $\max\{|y_i-\dot{\psi}(u(\tilde{\bm\beta}_0^T\bm x_i))|,\delta\}$ to replace $|y_i-\dot{\psi}(u(\tilde{\bm\beta}_0^T\bm x_i))|$. Here, $\delta$ is a small positive number, say $10^{-6}$ as an example. 
 Setting a threshold $\delta$ in subsampling probabilities results in a truncation in the weights for the subsample weighted log-likelihood.
 Truncating the weight function is a commonly used technique in practice for 
 robust estimation. 
 Note that in practical application, an intercept should always be included in a model, so it is typical that $\|{\dot u}(\hat{\bm\beta}_{\rm MLE}^T\bm x_i)\bm x_i\|$ and $\|\mathcal{J}_X^{-1}{\dot u}(\hat{\bm\beta}_{\rm MLE}^T\bm x_i)\bm x_i\|$ are bounded away from 0 and we do not need to set a threshold for them.
{Let $\tilde{V}$ be the version of $V$ with $\hat{\bm\beta}_{\rm MLE}$ substituted by $\tilde{\bm\beta}_0$. It can be shown that
\[\tr(\tilde{V})\le \tr(\tilde{V}^\delta)\le \tr(\tilde{V})+\frac{\delta^2}{n^2r}\sum_{i=1}^n\frac{1}{\pi_i}\|\tilde{\mathcal{J}}_X^{-1}{\dot u}(\hat{\bm\beta}_{\rm MLE}^T\bm x_i)\bm x_i\|^2. \]
Thus, minimizing $\tr(\tilde{V}^\delta)$ is close to minimizing $\tr(\tilde{V})$ if $\delta$ is sufficiently small. The threshold $\delta$ is
to make our subsampling estimator more robust without scarifying the estimation efficiency too much.
Here, we can also approximate $\mathcal{J}_X$ by using the pilot sample. 
To be specific, the $\mathcal{J}_X$ in mV is approximated by $\tilde{\mathcal{J}}_X=(r_0)^{-1}\sum_{i=1}^{r_0}{\{ \ddot u(} {\tilde{\bm\beta}^T}{\bm x_i^*}){\bm x_i^*}\bm x_i^{*T}[\dot \psi (u({\tilde{\bm\beta}^T}{\bm x_i}^*)) - {y_i}^*] + \ddot \psi (u({\tilde{\bm\beta}^T}{\bm x_i}^*)){\dot u^2}({\tilde{\bm\beta}^T}{\bm x_i}^*){\bm x_i}^*\bm x_i^{*T}]\}$ based on the first stage subsamples $\{(\bm x_i^*,y_i^*):i=1,\ldots,r_0\}$. }

For transparent presentation, we combine all the aforementioned practical considerations in this section and present a two-step algorithm as below.
 \begin{enumerate}
   \item Run the general subsampling algorithm with $\bm\pi=\bm\pi^{\rm UNIF}$ and $r=r_0$ to get the pilot subsample set $\tilde{S}_{r_0}$ and a pilot estimator $\tilde{\bm\beta}_0$.
   \item Using $\tilde{\bm\beta}_0$ to calculate approximated subsampling probabilities $\tilde{\bm\pi}^{\rm opt}=\{\tilde\pi_i^{\mathrm{mV}}\}_{i=1}^n$ or $\tilde{\bm\pi}^{\rm opt}=\{\tilde\pi_i^{\mathrm{mVc}}\}_{i=1}^n$, where $\tilde\pi_i^{\mathrm{mV}}$s are proportional to $\max(|y_i-\dot{\psi}(u(\tilde{\bm\beta}_0^T\bm x_i))|,\delta)\|\tilde{\mathcal{J}}_X^{-1}\dot{u}(\tilde{\bm\beta}_0^T\bm x_i)\bm x_i\|$s and $\tilde\pi_i^{\mathrm{mVc}}$s are proportional to $\max(|y_i-\dot{\psi}(u(\tilde{\bm\beta}_0^T\bm x_i))|,\delta)\|\dot{u}(\tilde{\bm\beta}_0^T\bm x_i)\bm x_i\|$.
   \item Sample with replacement for $r$ times based on $\tilde{\bm\pi}^{\rm opt}$ to form the subsample set $S_{r^*} := \tilde{S}_{r_0}\cup \{(y_{i}^*,\bm {x}_{i}^*,\tilde\pi_{i}^{*}),i=1,\ldots,r\}.$
   \item Maximize the following weighted log-likelihood function 
     to obtain the 
  estimator ${\breve{\bm\beta}}$
  \begin{equation}\label{eq:reweigt1}
      L^*(\bm\beta)=\frac{1}{r+r_0}\sum_{i\in S_{r^*}}\frac{1}{{\tilde\pi}_{i}^{*}}[y_{i}^*u({\bm\beta}^{T}\bm x_{i}^*)-\psi(u{\bm \beta}^{T}\bm x_{i}^*))].
  \end{equation}
 \end{enumerate}

 We have the following theorems describing asymptotic properties of $\breve{\bm\beta}$.
\begin{theorem}\label{thm:asy-2step-alg}
Under Assumptions~(H.1)--(H.5), if 
$r_0r^{-1}\rightarrow0$ as $r_0\rightarrow\infty, {r}\rightarrow\infty$ and $n\rightarrow\infty$, then for the estimator $\breve{\bm\beta}$ obtained from  the two-step algorithm,  
   with probability approaching one, for any $\epsilon>0$, there exist   finite $\Delta_\epsilon$ and $r_\epsilon$ such that
  \begin{equation*}
  P(\|{\breve{\bm\beta}}-\hat{\bm\beta}_{\rm MLE}\|\ge {r}^{-1/2}\Delta_\epsilon|\mathcal{F}_n)<\epsilon
\end{equation*}
for all ${r}>r_\epsilon$.
\end{theorem}

The result of asymptotic normality is presented in the following theorem.

{\begin{theorem}\label{thm:CLT-2step}
Under assumptions (H.1)--(H.5), if $r_0r^{-1}\rightarrow0$, then for the estimator obtained from the two-step algorithm, as $r_0\rightarrow\infty$, $r\rightarrow\infty$ and $n\rightarrow\infty$, conditional on $\mathcal{F}_n$,
  \begin{equation}\label{eq:clt-2}
    V_{opt}^{-1/2}(\breve{\bm\beta}-\hat{\bm\beta}_{\rm MLE}) \rightarrow N(0,I),
  \end{equation}
where $V_{opt}=\mathcal{J}_X^{-1}V_{c,opt}\mathcal{J}_X^{-1}$; and
  {\begin{align}
    V_{c,opt}&=\frac{1}{r}\onen
    \sum_{i=1}^n \frac{\{y_i-\dot{\psi}(u(\hat{\bm\beta}_{\rm MLE}^T\bm x_i))\}^2\dot{u}^2(\hat{\bm\beta}_{\rm MLE}^T\bm x_i)\bm x_i\bm x_i^T}{ {\max(|y_i-\dot{\psi}(u(\hat{\bm\beta}_{\rm MLE}^T\bm x_i))|,{\delta})\|\dot{u}(\hat{\bm\beta}_{\rm MLE}^T\bm x_i)\bm x_i\|}}\\
\nonumber  &  \times\onen\sumn\max(|y_i-\dot{\psi}(u(\hat{\bm\beta}_{\rm MLE}^T\bm x_i))|,{\delta})\|\dot{u}(\hat{\bm\beta}_{\rm MLE}^T\bm x_i)\bm x_i\|,
  \end{align}}
 for subsampling probabilities based on $\tilde{\pi}^{\mathrm{mVc}}_i$, and
  {\begin{align*}
    V_{c,opt}&=\frac{1}{r}\onen
    \sum_{i=1}^n \frac{\{y_i-\dot{\psi}(u(\hat{\bm\beta}_{\rm MLE}^T\bm x_i))\}^2\dot{u}^2(\hat{\bm\beta}_{\rm MLE}^T\bm x_i)\bm x_i\bm x_i^T}{ {\max(|y_i-\dot{\psi}(u(\hat{\bm\beta}_{\rm MLE}^T\bm x_i))|,{\delta})\|\mathcal{J}_X^{-1}\dot{u}(\hat{\bm\beta}_{\rm MLE}^T\bm x_i)\bm x_i\|}}\\
\nonumber  &  \times\onen\sumn\max(|y_i-\dot{\psi}(u(\hat{\bm\beta}_{\rm MLE}^T\bm x_i))|,{\delta})\|\mathcal{J}_X^{-1}\dot{u}(\hat{\bm\beta}_{\rm MLE}^T\bm x_i)\bm x_i\|.
  \end{align*}}
for subsampling probabilities based on $\tilde{\pi}^{\mathrm{mV}}_i$.
\end{theorem}}

In order to get standard error of the corresponding estimator, 
we estimate the variance-covariance matrix of $\breve{\bm\beta}$ by 
 $
   \breve{V}=\breve{\mathcal{J}}_X^{-1}\breve{V}_c\breve{\mathcal{J}}_X^{-1},
 $
where
\begin{align*}
  \breve{\mathcal{J}}_X
  &=\frac{1}{n(r_0+r)}\times\\
  &  \bigg\{\sum\limits_{i=1}^{r_0}  \frac{\ddot u(\breve{\bm\beta}^T\bm x_i^*){\bm x_i^*}\bm x_i^T[\dot \psi (u(\breve{\bm\beta}^T\bm x_i^*)) - {y_i^*}] + \ddot \psi (u(\breve{\bm\beta}^T\bm x_i^*)){\dot u^2}(\tilde{\bm\beta}_0^T{\bm x_i^*}){\bm x_i^*}\bm x_i^T}{{\pi_{i0}^*}}\\
&+\sum\limits_{s=1}^{r}  \frac{\ddot u(\breve{\bm\beta}^T\bm x_s^*){\bm x_s^*}\bm x_s^T[\dot \psi (u(\breve{\bm\beta}^T\bm x_s^*)) - {y_s^*}] + \ddot \psi (u(\breve{\bm\beta}^T\bm x_s^*)){\dot u^2}(\breve{\bm\beta}^T\bm x_s^*){\bm x_s}\bm x_s^T}{{\tilde\pi_s^*}}\bigg \},\\
\breve{V}_c &=\frac{1}{n^2(r_0+r)^2}\bigg\{\sum_{i=1}^{r_0}\frac{\{y_i-\dot{\psi}(u(\breve{\bm\beta}^T\bm x_i^*))\}^2\dot{u}^2(\breve{\bm\beta}^T\bm x_i^*)\bm x_i^*(\bm x_i^*)^T}{(\tilde\pi_{i0}^*)^2}\\
   & +\sum_{i=1}^{r}\frac{\{y_i^*-\dot{\psi}(u(\breve{\bm\beta}^T\bm x_i^*))\}^2\dot{u}^2(\breve{\bm\beta}^T\bm x_i^*)\bm x_i^*(\bm x_i^*)^T}{(\tilde\pi_i^*)^2}\bigg\},
\end{align*}
{$\pi_{i0}^*$'s are the subsampling probabilities used in the first stage}, and $\tilde\pi_i^*=\tilde\pi_i^{\mathrm{mV}*}$ or $\tilde\pi_i^{\mathrm{mVc}*}$ for $i=1,\ldots,r$.

\section{Numerical Studies}\label{sec:experiments}

\subsection{Simulation Studies}\label{sec:Poisson-sim}

In this section, we use simulation to evaluate the finite sample performance of the proposed method in Poisson regression and NBR.
 Computations are performed in {\verb"R"} \citep{Rpackage2018}.  The performance of a sampling strategy $\bm\pi$ is evaluated by the empirical mean squared error (eMSE) of the resultant estimator:
{ $\text{eMSE}=K^{-1}\sum_{k=1}^K\|{\bm\beta}_{\bm\pi}^{(k)}-\hat{\bm\beta}_{\rm MLE}\|,$}
where ${\bm\beta}_{\bm\pi}^{(k)}$ is the estimator from the $k$-th subsample with subsampling probability $\bm\pi$ and $\hat{\bm\beta}_{\rm MLE}$ is the MLE calculated from the whole dataset. We set  $K = 1000$ throughout this section.

\textbf{Poisson regression.}
Full data of size $n=10,000$ are generated from model $y|\bm x \sim \mathcal{P}(\exp(\bm\beta^{T}\bm x))$
with the true value of $\bm\beta$ 
being a $7\times1$ vector
of 0.5. We consider the following four cases to generate the covariates $\bm x_i=(x_{i1}, ..., x_{i7})^{T}$.
\begin{enumerate}[{Case} 1:]
\item The seven covariates are independent and identically distributed (i.i.d) from the standard uniform distribution, namely, $x_{ij}\overset{\text{i.i.d}}{\sim} U(\left[0,1 \right])$ for $j=1, ..., 7$.
\item The first two covariates are highly correlated. Specifically, $x_{ij}\overset{\text{i.i.d}}{\sim} U(\left[0,1 \right])$ for all $j$ except for $x_{i2}=x_{i1}+\varepsilon_{i}$ with $\varepsilon_{i}\overset{\text{i.i.d}}{\sim} U(\left[0,0.1 \right] )$. For this setup, the correlation coefficient between the first two covariates are about 0.8.
\item
This case is the same as the second one except that $\varepsilon_{i}$ $\overset{\text{i.i.d}}{\sim} U(\left[0,1 \right] )$. For this case, the correlation between the first two covariates is close to $0.5$.

\item This case is the same as the third one except that  $x_{ij}\overset{\text{i.i.d}}{\sim} U(\left[-1, 1 \right] )$ for $j=6, 7$. For this case, the bounds for each covariates are not all the same.
\end{enumerate}

We consider both $\tilde{\pi}_i^{\mathrm{mV}}$ and $\tilde{\pi}_i^{\mathrm{mVc}}$, 
 and choose the value of $\delta$ to be $\delta=10^{-6}$. 
 For comparison, we also consider uniform subsampling, i.e., $\pi_i=1/n$ for all $i$ and
 the leverage subsampling strategy in \cite{Ma2015A}
 in which $\pi_i=h_i/\sum_{j=1}^nh_i=h_i/p$ with $h_i=\bm x_i(\bm X^T\bm X)^{-1}\bm x_i$. Here $h_i$'s are the leverage scores for linear regression.  
 For GLMs, leverage scores are defined by using the adjusted covariate matrix, namely, $\tilde{h}_i=\tilde{\bm x}_i(\tilde{\bm X}^T\tilde{\bm X})^{-1}\tilde{\bm x}_i$, where $\tilde{\bm X}=(\tilde{\bm x}_1, ..., \tilde{\bm x}_n)^T$,
 $\tilde{\bm x}_i=\sqrt{-E\{\partial^2\log f(y_i|\tilde\theta_i)/\partial\theta^2\}}\bm x_i,$
 and  $\tilde\theta_i=\tilde{\bm\beta}^T\bm x_i$ with an initial estimate $\tilde{\bm\beta}_0$ 
 \citep[see][]{Lee1987Diagnostic}.
 In this example, simple algebra yields $\tilde{\bm x_i}=\sqrt{\exp({\tilde{\bm\beta}_0^T\bm x_i})}\bm x_i$. 
 For the leverage score subsampling, we considered both $h_i$ and $\tilde{h}_i$.  Here is a summary for the methods to be compared:
 UNIF, uniform subsample;
mV, $\pi_i=\tilde{\pi}_i^{\mathrm{mV}}$;
mVc, $\pi_i=\tilde{\pi}_i^{\mathrm{mVc}}$;
Lev, leverage sampling based on $h_i$;
Lev-A, adjust leverage sampling based on $\tilde{h}_i$.

We first consider the case with the first step sample size fixed. We let $r_0=200$, and second step sample size $r$ be 300, 500, 700, 1000, 1200 and 1400, respectively. For subsampling probabilities that do not depend on unknown parameters, 
they are  implemented with subsample size $r+r_0$ for fair comparisons.

\begin{figure}[!htp]
  \centering
  \begin{subfigure}{0.49\textwidth}
    \includegraphics[width=\textwidth]{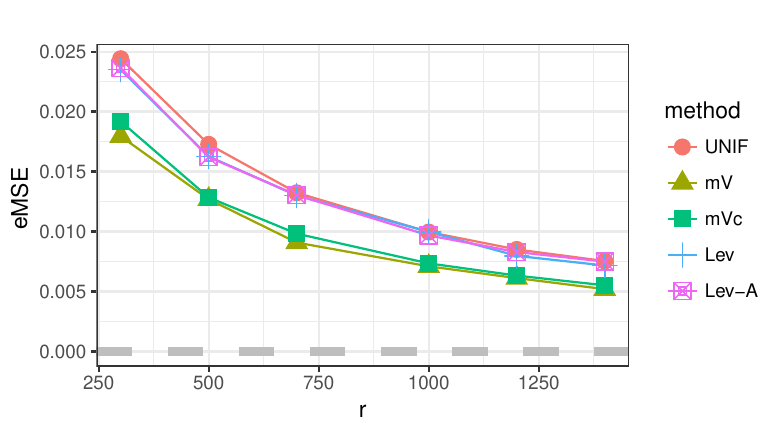}\\[-0.7cm]
    \caption{Case 1 (independent covariates)}
  \end{subfigure}
  \begin{subfigure}{0.49\textwidth}
    \includegraphics[width=\textwidth]{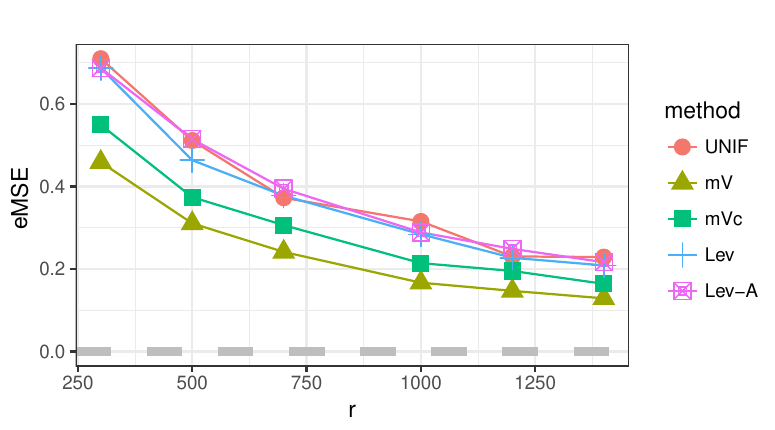}\\[-0.7cm]
    \caption{Case 2 (highly correlated covariates)}
  \end{subfigure}\\[5mm]
  \begin{subfigure}{0.49\textwidth}
    \includegraphics[width=\textwidth]{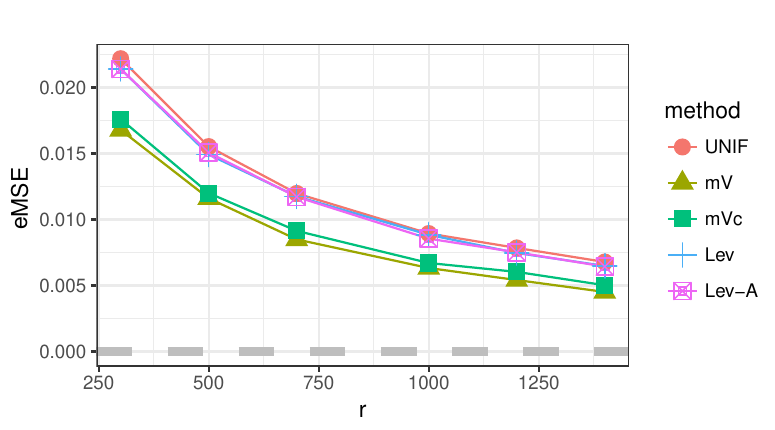}\\[-0.7cm]
    \caption{Case 3 (weakly correlated covariates)}
  \end{subfigure}
  \begin{subfigure}{0.49\textwidth}
    \includegraphics[width=\textwidth]{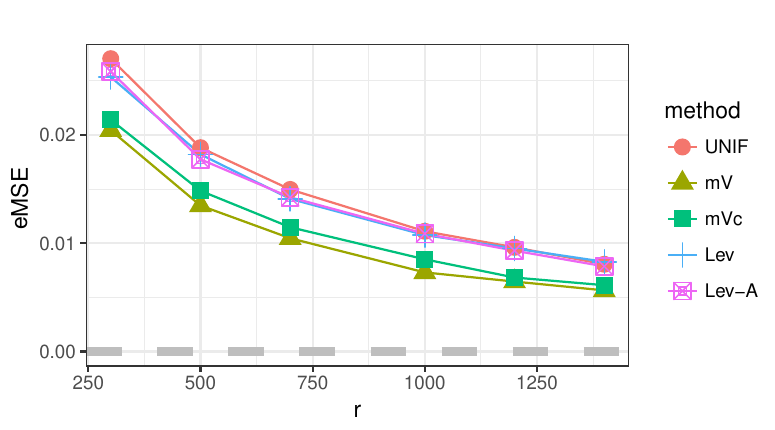}\\[-0.7cm]
    \caption{Case 4 (unequal bounds of covariates)}
  \end{subfigure}\\[5mm]
  \caption{The eMSEs for Poisson regression with different second step subsample size ${r}$ and a fixed first step subsample size $r_0=200$.
     The different distributions of covariates are listed in the beginning of Section 5.
   }
  \label{fig:3}
\end{figure}

Figure \ref{fig:3} gives 
  eMSEs.
 It is seen that for all the four data sets, subsampling methods based on {$\tilde{\bm\pi}^{\mathrm{mV}}$ and  $\tilde{\bm\pi}^{\mathrm{mVc}}$} always result in smaller eMSE than the uniform subsampling, which agrees with the theoretical result that they aim to minimize the asymptotic eMSEs of the resultant estimator. 
If the components of $\bm x$ are independent, $\tilde{\bm\pi}^{\mathrm{mV}}$ and  $\tilde{\bm\pi}^{\mathrm{mVc}}$ have similar performances, while they may perform differently if some covariates are highly correlated. The reason is that $\tilde{\bm\pi}^{\mathrm{mVc}}$ reduces the impact of the data correlation structure since $\|\tilde{\mathcal{J}}_X^{-1}\bm x_i\|^2$ in $\tilde{\bm\pi}^{\mathrm{mV}}$ are replaced by $\|\bm x_i\|^2$ in $\tilde{\bm\pi}^{\mathrm{mVc}}$.

 For Cases 1, 3 and 4, eMSEs are small. This is because the condition number of $\bm X_d$ is quite small ($\approx$ 5) and a small subsample size $r=100$ produces satisfactory results.
 However, for Case 2, the condition number is large ($\approx$ 40), so a larger subsample size is needed to approximate $\hat{\bm\beta}_{\rm MLE}$ accurately. This agrees with the conclusion in Theorem~\ref{thm:r-size}.

 Another contribution of Theorem \ref{thm:CLT-2step} is to enable us to do inference on $\bm\beta$. Note that in subsampling setting, $r$ is much smaller than the full data size $n$. If $r=o(n)$, then $\hat{\bm\beta}_{\rm MLE}$ in Theorem \ref{thm:CLT-2step} can be replaced by the true parameter. 
 As an example, we take $\beta_2$ as a parameter of interest and construct 95\% confidence intervals for it. For this, the estimator given by  $
   \breve{V}=\breve{\mathcal{J}}_X^{-1}\breve{V}_c\breve{\mathcal{J}}_X^{-1},
 $ is used to estimate variance-covariance matrices based on selected subsamples. For comparison, uniform subsampling method is also implemented.

 Table \ref{tab:tabl} reports empirical coverage probabilities and average lengths in Poisson regression model over the four synthetic data sets with 
 the first step subsample size being fixed at $r_0=200$.
It is clear that $\tilde{\bm\pi}^{\mathrm{mV}}$ and  $\tilde{\bm\pi}^{\mathrm{mVc}}$ have similar performances and are uniformly better than the uniform subsampling method. As $r$ increases,  the lengths of confidence intervals decrease uniformly which echos the results of Theorem \ref{thm:CLT-2step}. Confidence intervals in Case 2 are longer than those in other cases with the same subsample sizes. This is due to the fact that the condition number of $\bm X_d$ in Case 2 is bigger than that of $\bm X_d$ in other cases. This indicates that we should select a larger subsample when the condition number of the full dataset is bigger, which echoes the results discussed in Section \ref{sec:r-size}.

\begin{table}[htbp]
  \centering
  \caption{Empirical coverage probabilities and average lengths of confidence intervals for $\beta_2$. 
    The first step subsample size is fixed at $r_0=200$. }
    \setlength{\tabcolsep}{1mm}{
    \begin{tabular}{cccccccc}
    \hline
         & \multicolumn{1}{c}{method} & \multicolumn{2}{c}{mV} &  \multicolumn{2}{c}{mVc} &  \multicolumn{2}{c}{UNIF} \\

      & \multicolumn{1}{c}{r} & \multicolumn{1}{l}{Coverage} & \multicolumn{1}{l}{Length} & \multicolumn{1}{l}{Coverage} & \multicolumn{1}{l}{Length} &  \multicolumn{1}{l}{Coverage} & \multicolumn{1}{l}{Length} \\
    \hline
    \multirow{3}[2]{*}{case 1} & 300   & 0.954 & 0.2037    & 0.955 & 0.2066   & 0.952 & 0.2275  \\
          & 500   & 0.954 & 0.1684    & 0.945 & 0.1713   & 0.942 & 0.1924  \\
          & 1000  & 0.946 & 0.1254    & 0.938 & 0.1281    & 0.953 & 0.1471  \\
    \multirow{3}[2]{*}{case 2} & 300   & 0.961 & 1.9067    & 0.946 & 2.0776    & 0.950  & 2.2549  \\
          & 500   & 0.958 & 1.5470    & 0.948 & 1.7263   & 0.947 & 1.9082  \\
          & 1000  & 0.954 & 1.1379    & 0.948 & 1.2919    & 0.945 & 1.4559  \\
    \multirow{3}[2]{*}{case 3} & 300   & 0.959 & 0.1770    & 0.953  & 0.1816    & 0.939  & 0.2000  \\
          & 500   & 0.942 & 0.1451   & 0.949  & 0.1507    & 0.942  & 0.1693  \\
          & 1000  & 0.954 & 0.1082    & 0.954  & 0.1132    & 0.939  & 0.1291  \\
    \multirow{3}[2]{*}{case 4} & 300   & 0.955 & 0.2097   & 0.951 & 0.2179    & 0.953 & 0.2402  \\
          & 500   & 0.951 & 0.1721    & 0.956 & 0.1803    & 0.942 & 0.2033  \\
          & 1000  & 0.957 & 0.1276    & 0.960  & 0.1347    & 0.943 & 0.1552  \\
    \hline
    \end{tabular}}%
  \label{tab:tabl}%
\end{table}%

{
\textbf{Negative Binomial Regression}. We also perform simulation for the negative binomial regression with $n=100,000$ and summarize the results in Figure \ref{fig:3-nb}.
Here we assume $y_i|\bm x_i \sim \rm{NB}(\mu_i,\nu),\quad \mu_i=\exp(\bm\beta^T\bm x_i)$ with the size parameter $\nu=2$.
Other simulation settings are the same as the Poisson regression example.  
It is worthy to mention that compared with Poison regression the eMSE's are lager for NBR when $r$ is the same. This also coincides with Theorem \ref{thm:r-size} since $C_{\dot u}>1$ for NBR. 
Result for 95\% confidence intervals of $\beta_2$ are also reported in Table \ref{tab:tabl-nb}.

\begin{figure}[!htp]
  \centering
  \begin{subfigure}{0.49\textwidth}
    \includegraphics[width=\textwidth]{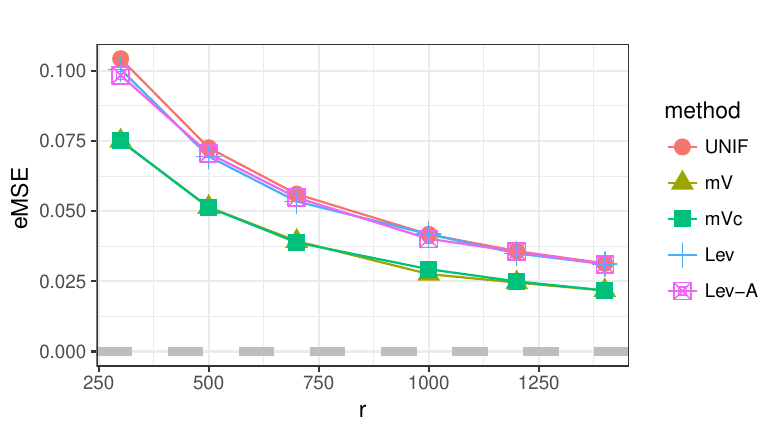}\\[-0.7cm]
    \caption{Case 1 (independent covariates)}
  \end{subfigure}
  \begin{subfigure}{0.49\textwidth}
    \includegraphics[width=\textwidth]{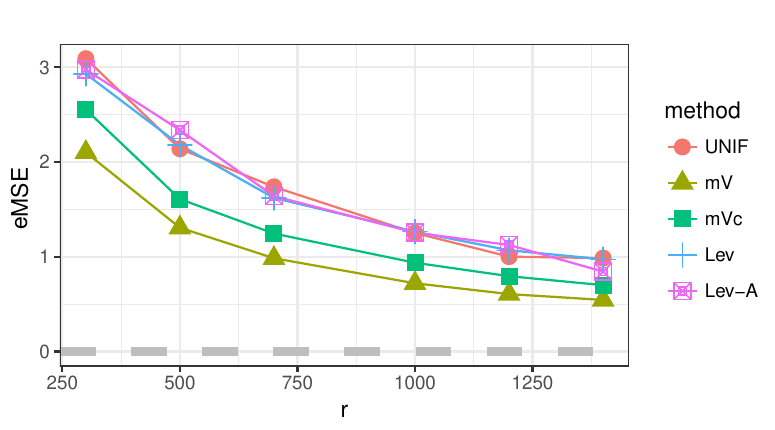}\\[-0.7cm]
    \caption{Case 2 (highly correlated covariates)}
  \end{subfigure}\\[5mm]
  \begin{subfigure}{0.49\textwidth}
    \includegraphics[width=\textwidth]{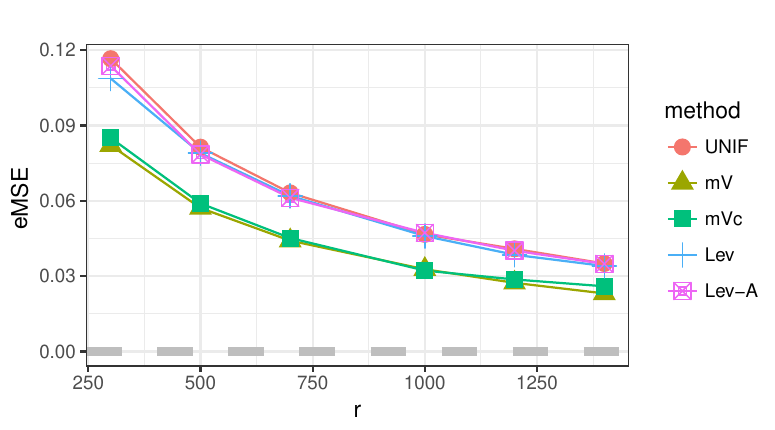}\\[-0.7cm]
    \caption{Case 3 (weakly correlated covariates)}
  \end{subfigure}
  \begin{subfigure}{0.49\textwidth}
    \includegraphics[width=\textwidth]{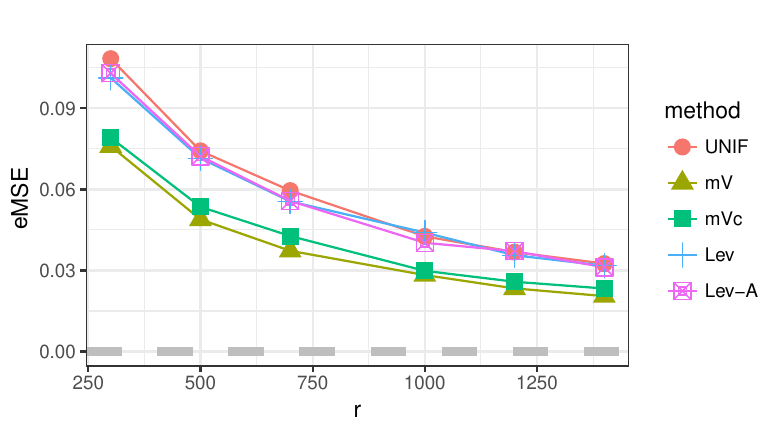}\\[-0.7cm]
    \caption{Case 4 (unequal bounds of covariates)}
  \end{subfigure}\\[5mm]
  \caption{The eMSEs for NBR with different second step subsample size ${r}$ and a fixed first step subsample size $r_0=200$.
     The different distributions of covariates are listed in the beginning of Section 5.}
  \label{fig:3-nb}
\end{figure}

\begin{table}[htbp]
  \centering
  \caption{Empirical coverage probabilities and average lengths of confidence intervals for 
    $\beta_2$ in NBR with $\nu=2$. 
    The first step subsample size is fixed at $r_0=200$.}
    \setlength{\tabcolsep}{1.5mm}{
    \begin{tabular}{cccccccc}
    \hline
          & method & \multicolumn{2}{c}{mV} & \multicolumn{2}{c}{mVc} & \multicolumn{2}{c}{UNIF} \\
     & \multicolumn{1}{c}{r} & \multicolumn{1}{l}{Coverage} & \multicolumn{1}{l}{Length} & \multicolumn{1}{l}{Coverage} & \multicolumn{1}{l}{Length} &  \multicolumn{1}{l}{Coverage} & \multicolumn{1}{l}{Length} \\
    \hline
    \multirow{3}[2]{*}{case1} & 300   & 0.952  & 0.2122  & 0.955  & 0.2147  & 0.947  & 0.2354  \\
          & 500   & 0.952  & 0.1758  & 0.954  & 0.1776  & 0.946  & 0.1991  \\
          & 1000  & 0.951  & 0.1305  & 0.933  & 0.1331  & 0.940  & 0.1520  \\
    \multirow{3}[2]{*}{case2} & 300   & 0.947  & 2.0228  & 0.963  & 2.2160  & 0.943  & 2.3913  \\
          & 500   & 0.953  & 1.6468  & 0.952  & 1.8423  & 0.946  & 2.0225  \\
          & 1000  & 0.957  & 1.2065  & 0.947  & 1.3849  & 0.942  & 1.5439  \\
    \multirow{3}[2]{*}{case3} & 300   & 0.950  & 0.1878  & 0.950  & 0.1925  & 0.942  & 0.2110  \\
          & 500   & 0.949  & 0.1546  & 0.954  & 0.1595  & 0.944  & 0.1786  \\
          & 1000  & 0.953  & 0.1150  & 0.957  & 0.1197  & 0.943  & 0.1361  \\
    \multirow{3}[2]{*}{case4} & 300   & 0.956  & 0.2288  & 0.953  & 0.2366  & 0.953  & 0.2573  \\
          & 500   & 0.968  & 0.1876  & 0.963  & 0.1956  & 0.936  & 0.2176  \\
          & 1000  & 0.950  & 0.1396  & 0.952  & 0.1469  & 0.940  & 0.1662  \\
    \hline
    \end{tabular}}%
  \label{tab:tabl-nb}%
\end{table}%

Now we investigate the effect of different sample size allocations between the two steps.
Since the results for Poisson and NBR have similar performances, we only report the results for Poisson regression to save space.
Here, we calculate eMSEs for various proportions of first step subsamples with fixed total subsample sizes. The results are given in Figure \ref{fig:4} with total subsample size $r_0+r$ = 800 and 1200, respectively. Since results are similar for all the cases, we only present results for Case 4 here.
It is worthy noting that the two-step method outperforms the uniform subsampling method for all the four cases for both Poisson and NBR,  when $r_0/r\in[0.1,0.9]$. 
This indicates that the two-step approach is more efficient than the uniform subsampling.  
The two-step approach works the best when $r_0/r$ is around 0.2.

\begin{figure}[!htp]
  \centering
  \begin{subfigure}{0.49\textwidth}
    \includegraphics[width=\textwidth]{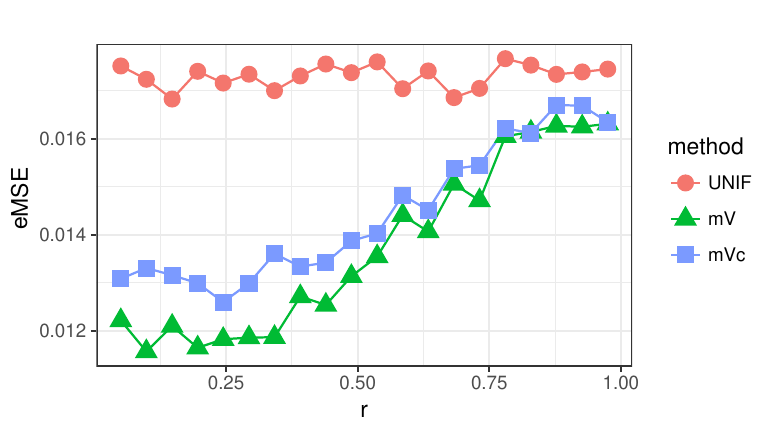}\\[-0.7cm]
    \caption{Case 4 ($r_0+r=800$)}
  \end{subfigure}
  \begin{subfigure}{0.49\textwidth}
    \includegraphics[width=\textwidth]{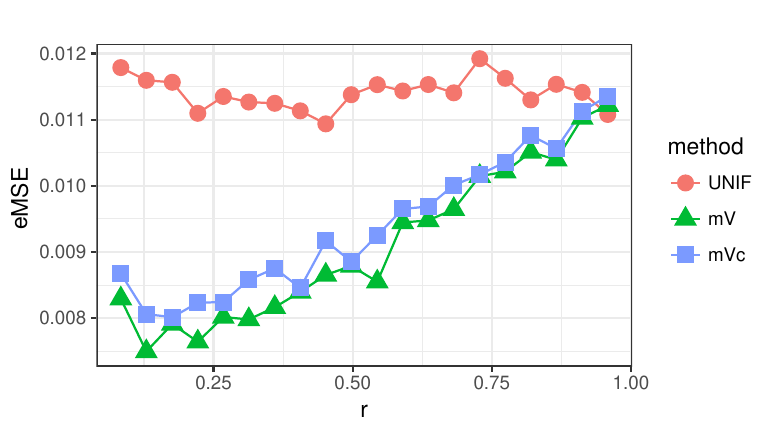}\\[-0.7cm]
    \caption{Case 4 ($r_0+r=1200$)}
  \end{subfigure}\\[5mm]
  \caption{The eMSEs vs proportions of the first step subsample with fixed
    total subsample sizes $r+r_0$ in Poisson regression.}
  \label{fig:4}
\end{figure}

{To explore the influence of $\delta$ in $\tilde{\pi}_i^{\mathrm{mV}}$ and $\tilde{\pi}_i^{\mathrm{mVc}}$, we calculate eMSEs for various $\delta$ ranging from $10^{-6}$ to 1  with fixed total subsample sizes.
Since the results for Poisson and NBR are similar, we only report the results for Poisson regression here.
Figure \ref{fig:q4} presents the results for Case 4 with total subsample size $r_0+r$ = 800 and 1200, respectively.
Judging from Figure \ref{fig:q4}, we see that the eMSE is not sensitive to the choice of $\delta$ when $\delta$ is not big, say $\delta=1$ for instance.
\begin{figure}[!htp]
  \centering
  \begin{subfigure}{0.49\textwidth}
    \includegraphics[width=\textwidth]{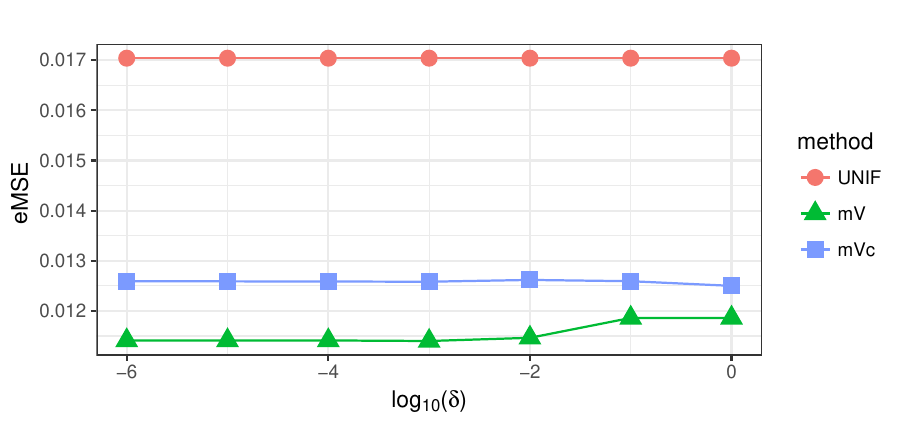}\\[-0.7cm]
    \caption{Case 4 ($r_0+r=800$)}
  \end{subfigure}
  \begin{subfigure}{0.49\textwidth}
    \includegraphics[width=\textwidth]{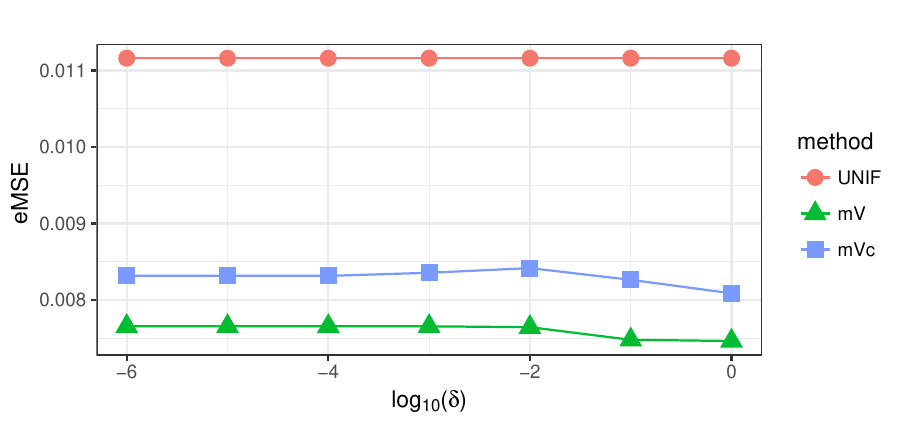}\\[-0.7cm]
    \caption{Case 4 ($r_0+r=1200$)}
  \end{subfigure}\\[5mm]
  \caption{The eMSEs vs  $\delta$ ranging from $10^{-6}$ to 1 with fixed
    total subsample sizes $r+r_0$ in Poisson regression. Logarithm is taken on $\delta$ for better presentation.}
  \label{fig:q4}
\end{figure}

To evaluate the computational efficiency of the subsampling
strategies, we record the computing times of the five subsampling strategies (uniform, $\pi^{\rm mV}$, $\pi^{\rm mVc}$, leverage score and adjust leverage score) by using \verb"Sys.time()" function in R  to record start and end times of the corresponding code.  Each subsampling strategy has been evaluated 50 times.
All methods are implemented with the R programming language.  
Computations were carried out on a
desktop running Window 10 with an Intel I7 processor and 32GB
memory.   Table~\ref{tab:tab2} shows the results
for Case 4 with different $r$ and a fixed $r_0=400$. The computing time for using the full data
is also given  for comparisons.

{
It is not surprising to observe  that the uniform subsampling algorithm requires the
least computing time because it does not require an additional step to
calculate the subsampling probability.
The algorithm based on $\bm{\pi}^{\mathrm{mV}}$ requires longer computing time than the algorithm based on
$\bm{\pi}^{\mathrm{mVc}}$, which agrees with the theoretical analysis in Section~\ref{sec:two-step}.
{The leverage score sampling takes nearly the same time as the mV method since leverage scores are computed directly by the definition.}
Note that $p=7$ is not big enough to use the fast computing method mentioned in \cite{Drineas2011Fast}.
For fairness, we also consider the case with $p=80,n=100,000$, which is suitable to use the fast computing method for the Lev and Lev-A methods.
The first seven variables are generated as Case 4 and the rest are generated independently from $U(\left[0,1 \right])$. Here $r_0$ is also selected as 400 and the corresponding results are reported in Table \ref{tab:tab3}.
In order to see the estimation effects we also present eMSEs in 
 {Tables \ref{tab:addtab1} and \ref{tab:addtab2}, respectively.}

{From Table \ref{tab:tab3} it is clear that all the subsampling algorithms take significantly less computing time compared to the full data approach. The Lev and Lev-A are faster than the mV method since the fast algorithm runs in $O(pn\log n)$ time to get the subsampling probabilities, as opposed to the $O(p^2n)$ time required by the mV method. However, the mVc method is still faster than Lev and Lev-A since the time complexity is just $O(pn)$ in computing the subsampling probabilities.
As the dimension increases, the computational  advantage of $\bm\pi^{\rm mVc}$ is even more significant.}

\begin{table}[htbp]
  \centering
  \caption{Computing  time (in second) for Poisson regression in Case 4 with  different $r$ and a fixed $r_0=400$.}
  \setlength{\tabcolsep}{3.5mm}{
    \begin{tabular}{ccccccc}
    \hline
     r     & FULL  & UNIF  & mV    & mVc   & Lev   & Lev-A \\
    \hline
    1000  & 0.187  & 0.003  & 0.020  & 0.016  & 0.024  & 0.031  \\
    1500  & 0.195  & 0.005  & 0.022  & 0.017  & 0.022  & 0.033  \\
    2000  & 0.193  & 0.007  & 0.021  & 0.018  & 0.026  & 0.036  \\
    2500  & 0.194  & 0.004  & 0.027  & 0.022  & 0.024  & 0.036  \\
    \hline
    \end{tabular}}%
  \label{tab:tab2}%
\end{table}%

\begin{table}[htbp]
  \centering
  \caption{Empirical MSE for Poisson regression demonstrated in Table \ref{tab:tab2}. The numbers in the parentheses are the standard errors.} 
  \scriptsize
    \begin{tabular}{cccccc}
    \hline
    \multicolumn{1}{c}{r} & \multicolumn{1}{c}{UNIF}       & \multicolumn{1}{c}{MV}       & \multicolumn{1}{c}{MVc} &        \multicolumn{1}{c}{Lev}        & \multicolumn{1}{c}{Lev-A}   \\
    \hline
    1000  & 0.0091  (0.0065)  & 0.0064  (0.0041)  & 0.0088  (0.0051)  & 0.0088  (0.0065)  & 0.0095  (0.0068)  \\
    1500  & 0.0071  (0.0054)  & 0.0047  (0.0034)  & 0.0049  (0.0038)  & 0.0067  (0.0049)  & 0.0070  (0.0051)  \\
    2000  & 0.0056  (0.0043)  & 0.0037  (0.0026)  & 0.0040  (0.0031)  & 0.0054  (0.0041)  & 0.0054  (0.0040)  \\
    2500  & 0.0045  (0.0032)  & 0.0030  (0.0021)  & 0.0033  (0.0025)  & 0.0044  (0.0034)  & 0.0047  (0.0036)  \\
    \hline
    \end{tabular}%
  \label{tab:addtab1}%
\end{table}%

\begin{table}[htbp]
  \centering
  \caption{Computing  time (in second) for Poisson regression with $n=100,000$, dimension $p=80$, different values of $r$, and a fixed $r_0=400$.}
  \setlength{\tabcolsep}{3.5mm}{
    \begin{tabular}{ccccccc}
    \hline
     r     & FULL  & UNIF  & mV   & mVc  & Lev   & Lev-A \\
    \hline
    1000  & 11.738  & 0.129  & 0.638  & 0.218  & 0.475  & 0.557  \\
    1500  & 11.659  & 0.163  & 0.689  & 0.253  & 0.514  & 0.595  \\
    2000  & 11.698  & 0.203  & 0.725  & 0.296  & 0.552  & 0.637  \\
    2500  & 12.005  & 0.240  & 0.777  & 0.339  & 0.602  & 0.681  \\
    \hline
    \end{tabular}}%
  \label{tab:tab3}%
\end{table}%

\begin{table}[htbp]
  \centering
  \caption{Empirical MSE for Poisson regression demonstrated in Table \ref{tab:tab3}. The numbers in the parentheses are the standard errors.}
\scriptsize
    \begin{tabular}{cccccc}
    \hline
    \multicolumn{1}{c}{r} & \multicolumn{1}{c}{UNIF}        & \multicolumn{1}{c}{MV}        & \multicolumn{1}{c}{MVc}        & \multicolumn{1}{c}{Lev}        & \multicolumn{1}{l}{Lev-A}   \\
    \hline
    1000  & 0.1003 (0.0174)  & 0.0786 (0.0135)  & 0.0782 (0.0136) & 0.1011 (0.0172)  & 0.1021 (0.0192) \\
    1500  & 0.0729 (0.0121) & 0.0582 (0.0100) & 0.0579 (0.0101) & 0.0722 (0.0125) & 0.0732 (0.0127) \\
    2000  & 0.0562 (0.0095) & 0.0472 (0.0085) & 0.0470 (0.0085) & 0.0565 (0.0094) & 0.0577 (0.0099) \\
    2500  & 0.0466 (0.0079)  & 0.0392 (0.0070)  & 0.0395 (0.0067) & 0.0463 (0.0078) & 0.0471 (0.0078) \\
    \hline
    \end{tabular}%
  \label{tab:addtab2}%
\end{table}%
}
}

\subsection{Real Data Studies} \label{sec:realdata}

In the following, we illustrate our methods described in Section \ref{sec:two-step} by applying them to a data set from musicology.
This data set contains 1,019,318 unique users' music play counts in the Echo Nest which is available at {http://labrosa.ee.columbia.edu/millionsong/tasteprofile}.
One of the challenge of this dataset is to build  a music recommendation system.
As a basic step, it is interesting to predict the play counts  by using the song information which has been collected in the Million Song Dataset \citep{Bertin-Mahieux2011}.
Since the major mode and minor mode usually express different feelings, the play counts may perform differently under this two modes. Thus we only focus on major mode in this example.
Besides the mode of music, the following six features are selected to describe the song characteristics: $x_1$, the duration of the track; $x_2$, the overall loudness of the song; $x_3$, tempo in BPM; $x_4$, artist hotttnesss; $x_5$, the song hotttnesss; $x_6$, the hottness of the album which is selected as the max value of the song hotttnesss in the album.
Here, $x_1$, $x_2$ and $x_3$ are features of a specified song; $x_4$, $x_5$ and $x_6$ are features of the artist, audience and  album respectively.
The last three features are subjective assessments by The Echo Nest, and all of them are on a scale between 0 and 1.
Since the first three variables in the data set are on different scales, we normalize them first.
In addition, we drop the \verb"NA" values in the dataset. After data cleaning, we have $n=205,032$ data points.
As a first attempt to capture the relationship between the play counts and all regressors described above, we fit the basic Poisson regression model and report the result in Figure \ref{fig:81}. 

Another way of modeling count data is to use NBR. 
For comparison, we also report the results from NBR in Figure \ref{fig:82}
with the size parameter select as $\theta=1.4$ which indicates over-dispersion of the data.

Similar to the case of synthetic data sets, we also compare our method with uniform subsampling and the leverage score subsampling methods, and report the results for $r$ varying from 600 to 2800.
The empirical  MSEs  are reported in Figure \ref{fig:8}. It is clear to that as $r$ increase, the eMSE decreases quickly for all methods. 
  Moreover,  $\bm{\pi}^{\mathrm{mV}}$ and $\bm{\pi}^{\mathrm{mVc}}$ perform similarly, and are uniformly better than the uniform subsampling and leverage score subsampling for larger values of $r$. It also worth to mention that the MSE in negative binomial regression is less than the Poisson regression. This is because the ratio of square Winsorized mean and Winsorized variance of $\bm y$ is around 1.4 which implies the data is over-dispersed. This echoes the results in Theorem \ref{thm:r-size} which advise us to include more subsamples for worse goodness of fitting. 

\begin{figure}[!htp]
  \centering
   \begin{subfigure}{0.49\textwidth}
    \includegraphics[width=\textwidth]{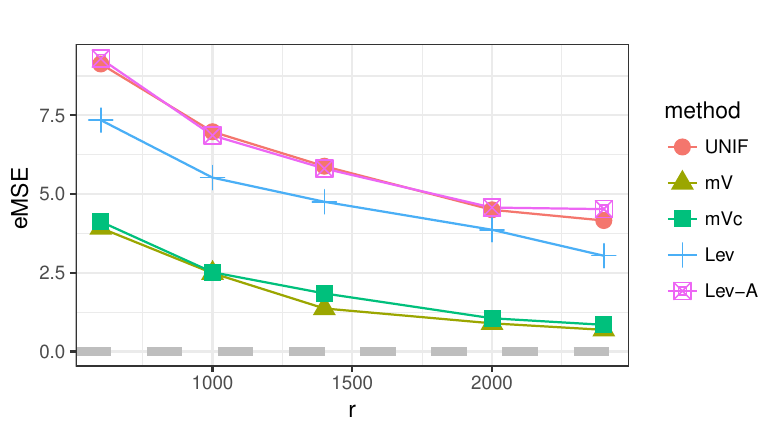}\\[-0.7cm]
    \caption{Poisson Regression}\label{fig:81}
  \end{subfigure}
  \begin{subfigure}{0.49\textwidth}
    \includegraphics[width=\textwidth]{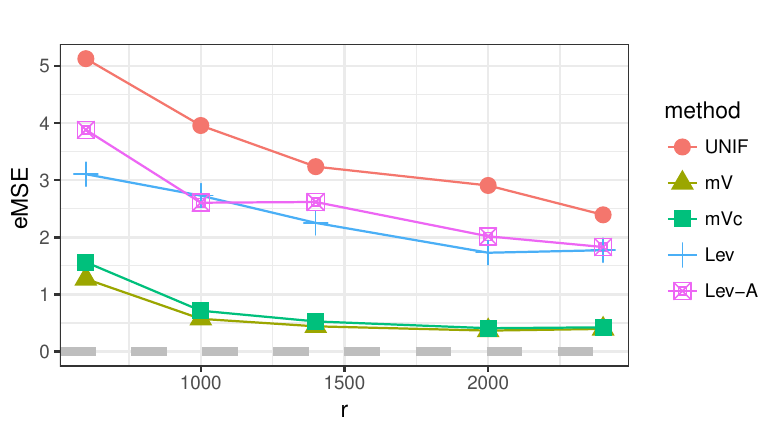}\\[-0.7cm]
    \caption{Negative Binomial Regression}\label{fig:82}
  \end{subfigure}\\[5mm]
  \caption{Empirical MSEs for different second step subsample size ${r}$ with
    the first step subsample size being fixed at $r_0=400$.}
  \label{fig:8}
\end{figure}

\section*{Supplementary Materials}
Technical proofs and additional simulation results are included in the online supplementary material.
\par

\section*{Acknowledgement}
{ The authors sincerely thank the editor, the associate editor and two referees for their valuable comments
which led to significant improvement of the manuscript.}
The authors would like to thank Prof. Jinzhu Jia for his helpful suggestions and discussions. Ai'work is supported by NNSF of China grants 11671019 and LMEQF. Wang's work is partially supported by NSF grant 1812013 and an UConn REP grant.

\par

\appendix
\section{Proofs}

{ To prove Theorem \ref{thm:as-general-alg}, we begin with the following remark and lemma.}

\begin{remark}\label{rm:analytic}
By the fact that $\lambda (\theta ): {=} \exp (\psi (\theta ))$ is analytic in {the interior of $\Theta$} (see Theorem 2.7 in \cite{brown1986fundamentals}), Cauchy's integral formula tells us that all its higher derivatives exist and are continuous. Therefore, the derivatives $\dot\psi(t)$,~$\ddot\psi(t)$,~$\dddot{\psi}(t)$ are continuous in $t$, and $\dot\psi(t)$,$\ddot\psi(t)$ are bounded on the compact set, which follows by a well-known property that every real-valued continuous function on a compact set is necessarily bounded.
\end{remark}

\begin{lemma}\label{lem:lem1}
If Assumptions (H.1)--(H.4) and (H.6) hold, then as $r\rightarrow\infty$ and $n\rightarrow\infty$, conditionally on $\mathcal{F}_n$ in probability,
\begin{align}
  \check{\mathcal{J}}_X-\mathcal{J}_X
  &=O_{P|\mathcal{F}_n}(r^{-1/2}),\label{eq:ob-infor}\\
  \frac{1}{n}L^*(\bm\beta)-\frac{1}{n}L(\bm\beta)
  &=O_{P|\mathcal{F}_n}(r^{-1/2}),\label{eq:L-L}\\
  \frac{1}{n}\frac{\partial L^*(\hat{\bm\beta}_{\rm MLE})}{\partial\bm\beta}
  &=O_{P|\mathcal{F}_n}(r^{-1/2}),\label{eq:ob-score}
\end{align}
where
\begin{align*}
  \check{\mathcal{J}}_X=
  -\frac{1}{n}\frac{\partial^2 L^*(\hat{\bm\beta}_{\rm MLE})}
  {\partial\bm\beta\partial\bm\beta^T}
  &=\frac{1}{nr}\sum_{i=1}^r\frac{\ddot{\psi}(u({{\bm\beta}^{T}{\bm x^*}_{i}}))
         \dot{u}({\bm\beta}^{T}{\bm x^*}_{i})\bm x_i^*[\dot{u}(\hat{\bm\beta}_{\rm MLE}^T{\bm x_{i}^*})\bm x_i^*]^T}{{\pi}_i^*}\\
         &+\frac{1}{{nr}}\sum\limits_{i = 1}^r {\frac{{\ddot u({\hat{\bm\beta}_{\rm MLE}^T{\bm x_{i}^*}}){\bm x_i^*}{\bm x^*}_i^T[\dot \psi (u({\hat{\bm\beta}_{\rm MLE}^T{\bm x}_{i}^*})) - {y_i^*}]}}{{\pi _i^*}}}.
\end{align*}
\end{lemma}
\begin{proof}
By the definition of conditional expectation and towering property of flirtations, it yields that
\[E(\check{\mathcal{J}}_X|\mathcal{F}_n)=\mathcal{J}_X.\]
For any component $\check{\mathcal{J}}_X^{j_1j_2}$ of $\check{\mathcal{J}}_X$ where $1\le j_1,j_2\le p$,
\begin{align}\label{eq:jxjx}
 & {\rm E}\left(\check{\mathcal{J}}_X^{j_1j_2}-\mathcal{J}_X^{j_1j_2}\Big|\mathcal{F}_n\right)^2\\
  =&\frac{1}{r}\sum^n_{i=1}\pi_i
     \bigg\{\frac{\ddot{\psi}(u(\hat{\bm\beta}_{\rm MLE}^T\bm x_i))\dot{u}^2(\hat{\bm\beta}_{\rm MLE}^T\bm x_i)x_{ij_1} x_{ij_2}}{n\pi_i}\nonumber\\
     &+{\frac{{\ddot u(\hat{\bm\beta}_{\rm MLE}^T{\bm x_i}){x_{i{j_1}}}{x_{i{j_2}}}[\dot \psi (u(\hat{\bm\beta}_{\rm MLE}^T{\bm x_i})) - {y_i}]}}{{n\pi _i}}}-\mathcal{J}_X^{j_1j_2}\bigg\}^2\nonumber\\
  = &\frac{1}{r}\sum^n_{i=1}\pi_i
     \bigg\{\frac{\ddot{\psi}(u(\hat{\bm\beta}_{\rm MLE}^T\bm x_i))\dot{u}^2(\hat{\bm\beta}_{\rm MLE}^T\bm x_i)x_{ij_1} x_{ij_2}}{n\pi_i}\nonumber\\
     &+{\frac{{\ddot u({\hat{\bm\beta}_{\rm MLE}^T}{\bm x_i}){x_{i{j_1}}}{x_{i{j_2}}}[\dot \psi (u({\hat{\bm\beta}_{\rm MLE}^T}{\bm x_i})) - {y_i}]}}{{n\pi _i}}}\bigg\}^2
     -\frac{1}{r}\left(\mathcal{J}_X^{j_1j_2}\right)^2\nonumber\\
  \le &\frac{2}{r}\sum^n_{i=1}\pi_i \bigg[
     \left\{\frac{\ddot{\psi}(u(\hat{\bm\beta}_{\rm MLE}^T\bm x_i))\dot{u}^2(\hat{\bm\beta}_{\rm MLE}^T\bm x_i)x_{ij_1} x_{ij_2}}{n\pi_i}\right\}^2\nonumber\\
     &+\left\{{\frac{{\ddot u({\hat{\bm\beta}_{\rm MLE}^T}{\bm x_i}){x_{i{j_1}}}{x_{i{j_2}}}[\dot \psi (u({\hat{\bm\beta}_{\rm MLE}^T}{\bm x_i})) - {y_i}]}}{{n\pi _i}}}\right\}^2\bigg]\nonumber\\
  \le &\frac{2}{r}\sum^n_{i=1}\pi_i \bigg[O_P(1)
     \left\{\frac{x_{ij_1} x_{ij_2}}{n\pi_i}\right\}^2 +O_P(1)\left\{{\frac{{{x_{i{j_1}}}{x_{i{j_2}}}[\dot \psi (u({\hat{\bm\beta}_{\rm MLE}^T}{\bm x_i})) - {y_i}]}}{{n\pi _i}}}\right\}^2\bigg]\nonumber\\
 =&{O_{P|{\mathcal{F}_n}}}(\frac{1}{r})\nonumber
\end{align}
where the last inequality stems form (H.1) and the last equality holds by assumptions (H.3) and (H.6). From Chebyshev's inequality, it is proved that Equation (\ref{eq:ob-infor}) holds.

To prove Equation (\ref{eq:ob-score}), let $t_i(\bm\beta)=y_i u(\bm\beta^T\bm x_i)-\psi(u(\bm\beta^T\bm x_i)$, $t_i^*(\bm\beta)=y_i^*u(\bm\beta^T\bm x_i^*)-\psi(u(\bm\beta^T\bm x_i^*))$, then
\begin{equation*}
L^*(\bm\beta)=\frac{1}{r}\sum^r_{i=1} \frac{t_i^*(\bm\beta)}{\tilde\pi_i^*}, \quad \text{and} \quad L(\bm\beta)=\sum^n_{i=1} t_i(\bm\beta).
\end{equation*}
Under the conditional distribution of the subsample given $\mathcal{F}_n$,
\[
E\left\{\frac{L^*(\bm\beta)}{n}-\frac{L(\bm\beta)}{n}\bigg|\mathcal{F}_n\right\}^2
    = \frac{1}{rn^2}\sum^n_{i=1}  \frac{t_i^2(\bm\beta)}{\pi_i}-\frac{1}{r}\left(\frac{1}{n}\sum^n_{i=1} t_i(\bm\beta)\right)^2.
\]
Combining the facts that the parameter space is compact and $u(t)$ is continuous function, by assumption (H.1) we have that $u(\bm\beta^T\bm x_i)$ are uniformly bounded. Then, it can be shown that by
Remark~\ref{rm:analytic}
{\normalsize\[\begin{split}
|t_i(\bm\beta)|&\le |y_i u(\bm\beta^T\bm x_i)-\dot{\psi}(u(\bm\beta^T\bm x_i))u(\bm\beta^T\bm x_i)|+|\dot{\psi}(u(\bm\beta^T\bm x_i))u(\bm\beta^T\bm x_i)-\psi(u(\bm\beta^T\bm x_i))|\\
&\le |[y_i-\dot{\psi}(u(\bm\beta^T\bm x_i))] u(\bm\beta^T\bm x_i)|+|\dot{\psi}(u(\bm\beta^T\bm x_i))u(\bm\beta^T\bm x_i)|+|\psi(u(\bm\beta^T\bm x_i))|,
\end{split}\]}

Therefore, 
we have
\[\begin{split}
\left(\frac{1}{n}\sum^n_{i=1} t_i(\bm\beta)\right)^2
&\le \left(\frac{1}{n}\sum^n_{i=1} O_P(1)|y_i-\dot{\psi}(u(\bm\beta^T\bm x_i))|\right)^2 \\
&\quad\quad+ \frac{O_P(1)}{n}\sum^n_{i=1}|y_i-\dot{\psi}(u(\bm\beta^T\bm x_i))| +O_P(1).
\end{split} \]
From Assumptions (H.1), we have $\mathop {\sup}\limits_n n^{-1}\sum^n_{i=1} t_i(\bm\beta) < \infty $.
Thus,
\begin{equation}\label{eq:l2error}
E\left\{\frac{L^*(\bm\beta)}{n}-\frac{L(\bm\beta)}{n}\bigg|\mathcal{F}_n\right\}^2=O_{P|\mathcal{F}_n}(r^{-1/2}).
\end{equation}
Now the desired result (\ref{eq:L-L}) follows from Chebyshev's Inequality.

Similarly, we can show that
\[\text{Var}\left(\frac{1}{n}\frac{\partial L^*(\hat{\bm\beta}_{\rm MLE})}{\partial\bm\beta}\right)=O_P(r^{-1}).\]
Thus (\ref{eq:ob-score}) is true.
\end{proof}

\subsection{Proof of {Theorem~\ref{thm:as-general-alg}}}
\begin{proof}
As $r\rightarrow\infty$, by (\ref{eq:l2error}), we have that $n^{-1}L^*(\bm\beta)-n^{-1}L(\bm\beta)\rightarrow0$ in conditional probability given $\mathcal{F}_n$. Note that the parameter space is compact and $\hat{\bm\beta}_{\rm MLE}$ is the unique global maximum of the continuous convex function $L(\bm\beta)$. Thus, from Theorem 5.9 and its remark of \cite{Vaart2000Asymptotic}, by \eqref{eq:ob-score} we have
\begin{equation}\label{eq:12}
  \|\tilde{\bm\beta}-\hat{\bm\beta}_{\rm MLE}\|=o_{P|\mathcal{F}_n}(1).
\end{equation}
as $n\rightarrow\infty, r\rightarrow\infty$, conditionally on $\mathcal{F}_n$ in probability.

Using Taylor's theorem for random variables \citep[see][Chapter 4]{Ferguson1996A},
\begin{align}
0=\frac{\dot{L}^*_j(\tilde{\bm\beta})}{n}
=&\frac{\dot{L}^*_j(\hat{\bm\beta}_{\rm MLE})}{n}
   +\frac{1}{n} \frac{\partial\dot{L}^*_j(\hat{\bm\beta}_{\rm MLE})}{\partial \bm\beta^T}(\tilde{\bm\beta}-\hat{\bm\beta}_{\rm MLE})
   +\frac{1}{n} R_j,
   \label{eq:taylor-expansion}
\end{align}
where $\dot{L}^*_j({\bm\beta})$ is the partial derivative of $L^*({\bm\beta})$ with respect to $\beta_j$, and the 
remainder
{\normalsize\begin{equation*}
  \frac{1}{n}R_j=\frac{1}{n}(\tilde{\bm\beta}-\hat{\bm\beta}_{\rm MLE})^T
   \int_0^1\int_0^1\frac{\partial^2\dot{L}^*_j\{\hat{\bm\beta}_{\rm MLE}
   +uv(\tilde{\bm\beta}-\hat{\bm\beta}_{\rm MLE})\}}{\partial \bm\beta\partial
   \bm\beta^T}v du dv\
(\tilde{\bm\beta}-\hat{\bm\beta}_{\rm MLE}).
\end{equation*}}
By calculus, we get
{\normalsize\begin{align*}
\frac{\partial^2\dot{L}^*_j(\bm\beta)}
  {\partial\bm\beta\partial\bm\beta^T}&=\frac{1}{r}\sum\limits_{i = 1}^r {\{ \frac{{\dddot{u}({\bm\beta ^T}{\bm x_i^*})[{y_i^*} - \dot \psi (u({\bm\beta ^T}{\bm x_i^*}))]}}{{\pi _i^*}}}  - \frac{{\ddot u({\bm \beta ^T}{\bm x_i^*})\dot u({\beta ^T}{\bm x_i^*})\dot \psi (u({\bm \beta ^T}{\bm x_i^*}))}}{{\pi _i^*}}\\
  & - \frac{{2\ddot \psi (u({\bm \beta ^T}{\bm x_i^*}))\ddot u({\bm \beta ^T}{\bm x_i^*})\dot u({\bm \beta ^T}{\bm x_i^*})}}{{\pi _i^*}} - \frac{{\dddot{\psi}(u({\bm \beta ^T}{\bm x_i^*})){{\dot u}^2}({\bm \beta ^T}{\bm x_i^*})}}{{\pi _i^*}}\} {x_{ij}^*}{\bm x_i^*}{\bm x_i^*}^T.
\end{align*}}
From (H.1) and Remark~\ref{rm:analytic}, we have
\begin{align*}
 & \frac{1}{n}\left\|\frac{\partial^2\dot{L}^*_j(\bm\beta)}
  {\partial\bm\beta\partial\bm\beta^T}\right\|_S\\
 =&\frac{1}{nr}\bigg\|\sum_{i=1}^r \bigg(  \frac{{\ddot u({\bm\beta ^T}{\bm x^*_i})\dot \psi (u({\bm\beta ^T}{\bm x^*_i}))}}{{\pi _i^*}} + \frac{{2\ddot \psi (u({\bm \beta ^T}{\bm x^*_i}))\ddot u({\bm \beta ^T}{\bm x^*_i})}}{{\pi _i^*}}+ \\
 &\frac{{\dddot{\psi} (u({\bm \beta ^T}{\bm x^*_i}))\dot u({\bm \beta ^T}{\bm x^*_i})}}{{\pi _i^*}}\bigg)
 {\dot u({\bm \beta ^T}{\bm x^*_i})}x^*_{ij}\bm x^*_i{\bm x^*_i}^T\bigg\|_S\\
  &+\frac{1}{rn}{\left\| {\sum\limits_{i = 1}^r {\frac{{\dddot{u}({\bm \beta ^T}{\bm x^*_i})[{y^*_i} - \dot \psi (u({\bm \beta ^T}{\bm x^*_i}))]}}{{\pi _i^*}}} } \right\|_S}\\
 \le &\frac{C_3}{rn}\sum_{i=1}^r\frac{\|\bm x^*_i\|^3}{\pi^*_i}+\frac{C_4}{rn}{\left\| {\sum\limits_{i = 1}^r {\frac{{|[{y^*_i} - \dot \psi (u({\bm \beta ^T}{\bm x^*_i})]x^*_{ij}|}}{{\pi _i^*}}} } \bm x^*_i{\bm x^*_i}^T\right\|_S}.
\end{align*}
for all $\bm\beta \in \Lambda _B$, where $C_3$ and $C_4$ are some constants by using Remark~\ref{rm:analytic}.

As $\tau\rightarrow\infty$, assumption (H.5) gives,
\begin{align*}
  P\left(\frac{1}{nr}\sum_{i=1}^r\frac{\|\bm x^*_i\|^3}{\pi^*_i}\ge\tau
  \Bigg|\mathcal{F}_n\right) \le\frac{1}{nr\tau}\sum_{i=1}^r
  E\left(\frac{\|\bm x^*_i\|^3}{\pi^*_i}\Bigg|\mathcal{F}_n\right)
  =\frac{1}{n\tau}\sum_{i=1}^n\|\bm x_i\|^3\overset{P}{\rightarrow} 0.
\end{align*}

Also note that as $\tau\to\infty$,
\begin{align*}
& P\left( {\frac{1}{{nr}}\sum\limits_{i = 1}^r {{{\left\| {\frac{{|[{y_i} - \dot \psi (u({\bm \beta ^T}{\bm x_i}))]{x_{ij}}|{\bm x_i}\bm x_i^T}}{{\pi _i^*}}} \right\|}_S}}  \ge \tau \Big|{{\cal F}_n}} \right)\\
& \le \frac{1}{{nr\tau }}E\left( {\sum\limits_{i = 1}^r {{{\left\| {\frac{{|[{y_i} - \dot \psi (u({\bm \beta ^T}{\bm x_i}))]{x_{ij}}|{\bm x_i}\bm x_i^T}}{{\pi _i^*}}} \right\|}_S}} \Big|{{\cal F}_n}} \right)\\
& \le \frac{1}{{n\tau }}\sum\limits_{i = 1}^n {{{\left\| {{{|[{y_i} - \dot \psi (u({\bm\beta ^T}{\bm x_i}))]{x_{ij}}|{\bm x_i}\bm x_i^T}}} \right\|}_S}}\\
&  \le \frac{1}{{\tau n}}\sum\limits_{i = 1}^n {|[{y_i} - \dot \psi (u({\bm \beta ^T}{\bm x_i}))]|}  \times {\left\| {{x_{ij}}{\bm x_i}\bm x_i^T} \right\|_S}\\
& \le \frac{1}{\tau }\sqrt {\frac{1}{n}\sum\limits_{i = 1}^n {|{y_i} - \dot \psi (u({\bm \beta ^T}{\bm x_i})){|^2}} }  \cdot \sqrt {\frac{1}{n}\sum\limits_{i = 1}^n {\left\| {{x_{ij}}{\bm x_i}\bm x_i^T} \right\|_S^2} }\\
& {\rm{ = }}\frac{1}{\tau }{O_P}(1) \overset{P}{\to} 0,
\end{align*}
where the last equality is due to (H.3) and (H.5) by noticing that
\begin{align*}
&\frac{1}{n}\sum\limits_{i = 1}^n {\left\| {{x_{ij}^*}{\bm x_i^*}\bm x_i^{*T}} \right\|_S^2}  \le \frac{1}{n}\sum\limits_{i = 1}^n {{{\left| {{x_{ij}^*}} \right|}^2}\left\| {{\bm x_i^*}\bm x_i^{*T}} \right\|_S^2}\\
\le& \frac{1}{n}\sum\limits_{i = 1}^n {{{\left\| {{\bm x_i}} \right\|}^2}\left\| {{\bm x_i}\bm x_i^T} \right\|_S^2}
  \le \frac{1}{n}\sum\limits_{i = 1}^n {{{\left\| {{\bm x_i}} \right\|}^6}} .
\end{align*}
Thus we have
\[\frac{1}{n}\left\|\frac{\partial^2\dot{L}^*_j(\bm\beta)}
  {\partial\bm\beta\partial\bm\beta^T}\right\|_S=O_{P|\mathcal{F}_n}(1).\]
From (H.1)-(H.3) and Remark~\ref{rm:analytic}, it is known that 
\begin{align*}
&  \frac{1}{n}\left\|\int_0^1\int_0^1\frac{\partial^2\dot{L}^*_j\{\hat{\bm\beta}_{\rm MLE}
   +uv(\tilde{\bm\beta}-\hat{\bm\beta}_{\rm MLE})\}}{\partial \bm\beta\partial\bm\beta^T}vdudv\ \right\|\\
& \le \frac{1}{n}\int_0^1\int_0^1\left\|\frac{\partial^2\dot{L}^*_j\{\hat{\bm\beta}_{\rm MLE}
   +uv(\tilde{\bm\beta}-\hat{\bm\beta}_{\rm MLE})\}}{\partial \bm\beta\partial\bm\beta^T}\right\| vdudv =O_{P|\mathcal{F}_n}(1).
\end{align*}

Combining the above equations with the Taylor's expansion (\ref{eq:taylor-expansion}), we have
\begin{equation}\label{eq:13-2}
  \tilde{\bm\beta}-\hat{\bm\beta}_{\rm MLE}=
  -\check{\mathcal{J}}_X^{-1}\left\{\frac{\dot {L}^*(\hat{\bm\beta}_{\rm MLE})}{n}
    +O_{P|\mathcal{F}_n}(\|\tilde{\bm\beta}-\hat{\bm\beta}_{\rm MLE}\|^2)\right\}.
\end{equation}
From  Lemma~\ref{lem:lem1} and Assumption (H.4), it is obvious that $\check{\mathcal{J}}_X^{-1} =O_{P|\mathcal{F}_n}(1)$. Therefore,
\begin{equation*}
  \tilde{\bm\beta}-\hat{\bm\beta}_{\rm MLE}
  =O_{P|\mathcal{F}_n}(r^{-1/2})+o_{P|\mathcal{F}_n}(\|\tilde{\bm\beta}-\hat{\bm\beta}_{\rm MLE}\|),
\end{equation*}
which implies that
$\tilde{\bm\beta}-\hat{\bm\beta}_{\rm MLE}=O_{P|\mathcal{F}_n}(r^{-1/2})$.
\end{proof}

\subsection{Proof of {Theorem~\ref{thm:CLT}}}
\begin{proof}
  Note that
\begin{equation}\label{eq:linder1}
  \frac{\dot L^*(\hat{\bm\beta}_{\rm MLE})}{n}
  =\frac{1}{r}\sum^r_{i=1}\frac{\{y^*_i-\dot{\psi}(u(\hat{\bm\beta}_{\rm MLE}^T\bm x_i^*))\}\dot{u}(\hat{\bm\beta}_{\rm MLE}^T\bm x_i^*)\bm x^*_i}{n\pi^*_i}
  =:\frac{1}{r}\sum^r_{i=1}\eta_i.
\end{equation}
It can be seen that given $\mathcal{F}_n$, $\eta_1, \ldots, \eta_r$ are i.i.d  random variables with mean $\mathbf{0}$ and variance
\begin{align}\label{eq:linder2}
  &\text{var}(\eta_1|\mathcal{F}_n)=
 \frac{1}{n^2}\sum^n_{i=1}\frac{\{y_i-\dot{\psi}_i(u(\hat{\bm\beta}_{\rm MLE}^T\bm x_i))\}^2
\dot{u}^2(\hat{\bm\beta}_{\rm MLE}^T\bm x_i) \bm x_i\bm x_i^T}{\pi_i}. 
\end{align}
Then from (H.7) with $\gamma=0$, we know that $\text{var}(\eta_i|\mathcal{F}_n)=O_P(1)$ as $n\rightarrow\infty$.

Meanwhile, for some $\gamma>0$ and every $\varepsilon>0$,
\begin{equation}\label{eq:dbclt}
\begin{aligned}
&\sum^r_{i=1} E\{\|r^{-1/2}\eta_i\|^2
    I(\|\eta_i\|>r^{1/2}\varepsilon)|\mathcal{F}_n\}\\
    \le&\frac{1}{r^{1+\gamma/2}\varepsilon^{\gamma}}
    \sum^r_{i=1}E\{\|\eta_i\|^{2+\gamma}
    I(\|\eta_i\|>r^{1/2}\varepsilon)|\mathcal{F}_n\}\\
\le&\frac{1}{r^{1+\gamma/2}\varepsilon^{\gamma}}
    \sum^r_{i=1} E(\|\eta_i\|^{2+\gamma}|\mathcal{F}_n)\\
=& \frac{1}{r^{\gamma/2}}\frac{1}{\varepsilon^{\gamma}}\frac{1}{n^{2+\gamma}}
    \sum^n_{i=1}\frac{|y_i-\dot{\psi}_i(u(\hat{\bm\beta}_{\rm MLE}^T\bm x_i))|^{2+\gamma}
    \|\dot{u}(\hat{\bm\beta}_{\rm MLE}^T\bm x_i)\bm x_i\|^{2+\gamma}}{\pi_i^{1+\gamma}}.
\end{aligned}
\end{equation}
From (H.7) for some $\gamma>0$, we obtain
\begin{align*}
\sum^r_{i=1} E\{\|r^{-1/2}\eta_i\|^2
    I(\|\eta_i\|>r^{1/2}\varepsilon)|\mathcal{F}_n\} \le \frac{1}{r^{\gamma/2}}\frac{1}{\varepsilon^{\gamma}}
O_P(1)\cdot O_P(1)=o_P(1),
\end{align*}
This shows that the Lindeberg-Feller conditions are satisfied in probability.
From (\ref{eq:linder1}) and (\ref{eq:linder2}), by the Lindeberg-Feller central limit theorem \citep[Proposition 2.27 of][]{Vaart2000Asymptotic}, conditionally on $\mathcal{F}_n$, 
\begin{equation*}
  \frac{1}{n}V_c^{-1/2}\dot{L}^*(\hat{\bm\beta}_{\rm MLE})=
  \frac{1}{r^{1/2}}\{\text{var}(\eta_i|\mathcal{F}_n)\}^{-1/2}\sum^r_{i=1}\eta_i
  \rightarrow N(0,{I}),
\end{equation*}
in distribution.
From Lemma~\ref{lem:lem1}, (\ref{eq:13-2}) and Theorem \ref{thm:as-general-alg}, we have
\begin{equation}\label{eq:15}
  \tilde{\bm\beta}-\hat{\bm\beta}_{\rm MLE}=
  -\frac{1}{n}\check{\mathcal{J}}_X^{-1}\dot{L}^*(\hat{\bm\beta}_{\rm MLE})+O_{P|\mathcal{F}_n}(r^{-1}).
\end{equation}
From (\ref{eq:ob-infor}) in Lemma~\ref{lem:lem1}, it follows that
\begin{align}\label{eq:7}
  \check{\mathcal{J}}_X^{-1}-\mathcal{J}_X^{-1}
  &=-\mathcal{J}_X^{-1}(\check{\mathcal{J}}_X-\mathcal{J}_X)\check{\mathcal{J}}_X^{-1}
  =O_{P|\mathcal{F}_n}(r^{-1/2}).
\end{align}
Based on Assumption (H.4) and (\ref{eq:linder2}), it can be proved that
\begin{equation*}
  V={\mathcal{J}}_X^{-1}V_c{\mathcal{J}}_X^{-1}
  =\frac{1}{r}{\mathcal{J}}_X^{-1}\left(rV_c\right)
  {\mathcal{J}}_X^{-1}=O_{P}(r^{-1}).
\end{equation*}
Thus,
\begin{align*}
  &V^{-1/2}(\tilde{\bm\beta}-\hat{\bm\beta}_{\rm MLE})
=-V^{-1/2}n^{-1}\check{\mathcal{J}}_X^{-1}\dot{L}^*(\hat{\bm\beta}_{\rm MLE})
    +O_{P|\mathcal{F}_n}(r^{-1/2})\\
=&-V^{-1/2}\mathcal{J}_X^{-1}n^{-1}\dot{L}^*(\hat{\bm\beta}_{\rm MLE})
    -V^{-1/2}(\check{\mathcal{J}}_X^{-1}-\mathcal{J}_X^{-1})n^{-1}\dot{L}^*(\hat{\bm\beta}_{\rm MLE})
    +O_{P|\mathcal{F}_n}(r^{-1/2})\\
=&-V^{-1/2}\mathcal{J}_X^{-1}V_c^{1/2}V_c^{-1/2}n^{-1}\dot{L}^*(\hat{\bm\beta}_{\rm MLE})
    +O_{P|\mathcal{F}_n}(r^{-1/2}).
\end{align*}
So the result in \eqref{normal} of Theorem~\ref{thm:CLT} follows by applying Slutsky's Theorem \citep[Theorem 6, Section 6 of][]{Ferguson1996A} and the fact that
\begin{equation*}
  V^{-1/2}\mathcal{J}_X^{-1}V_c^{1/2}(V^{-1/2}\mathcal{J}_X^{-1}V_c^{1/2})^T
  =V^{-1/2}\mathcal{J}_X^{-1}V_c^{1/2}V_c^{1/2}\mathcal{J}_X^{-1}V^{-1/2}={I}.
\end{equation*}
\end{proof}

\subsection{Proof of {Theorem~\ref{thm:3}}}

\begin{proof}
Note that
\begin{align*}
  &\text{tr}({V})=\text{tr}(\mathcal{J}_X^{-1}{V}_c\mathcal{J}_X^{-1})\\
=&\frac{1}{n^2r}\sum_{i=1}^n\text{tr}\left[\frac{1}{\pi_i}
    \{y_i-\dot{\psi}(u(\hat{\bm\beta}_{\rm MLE}^T\bm x_i))\}^2\mathcal{J}_X^{-1}\dot{u}(\hat{\bm\beta}_{\rm MLE}^T\bm x_i)\bm x_i[\dot{u}(\hat{\bm\beta}_{\rm MLE}^T\bm x_i)\bm x_i]^T\mathcal{J}_X^{-1}\right]\\
=&\frac{1}{n^2r}\sum_{i=1}^n\left[\frac{1}{\pi_i}\{y_i-\dot{\psi}(u(\hat{\bm\beta}_{\rm MLE}^T\bm x_i))\}^2
    \|\mathcal{J}_X^{-1}\dot{u}(\hat{\bm\beta}_{\rm MLE}^T\bm x_i)\bm x_i\|^2\right]\\
=&\frac{1}{n^2r}(\sum_{i=1}^n\pi_i)\sum_{i=1}^n\left[\pi_i^{-1}\{y_i-\dot{\psi}(u(\hat{\bm\beta}_{\rm MLE}^T\bm x_i))\}^2
    \|\mathcal{J}_X^{-1}\dot{u}({\bm\beta}^{T}{\bm x}_{i})\bm x_i\|^2\right]\\
\ge&\frac{1}{n^2r}\left[\sum_{i=1}^n|y_i-\dot{\psi}(u(\hat{\bm\beta}_{\rm MLE}^T\bm x_i))|\|\mathcal{J}_X^{-1}\dot{u}(\hat{\bm\beta}_{\rm MLE}^T\bm x_i)\bm x_i\|\right]^2,
\end{align*}
where the last inequality follows from the Cauchy-Schwarz inequality, and the equality in it holds if and only if $$\pi_i\propto |y_i-\dot{\psi}(u(\hat{\bm\beta}_{\rm MLE}^T\bm x_i))|\|\mathcal{J}_X^{-1}\bm x_i\|I\{|y_i-\dot{\psi}(u(\hat{\bm\beta}_{\rm MLE}^T\bm x_i))|\|\mathcal{J}_X^{-1}\bm x_i\|>0\}.$$
 Here we define $0/0=0$, and this is equivalent to removing data points with $|y_i-\dot{\psi}(u(\hat{\bm\beta}_{\rm MLE}^T\bm x_i))|=0$ in the expression of $V_c$.
\end{proof}

\subsection{Proof of {Theorems~\ref{thm:r-size} and~\ref{thm:r-size1}}}

 Let ${\left\| A \right\|_F}: =  (\sum_{i = 1}^m {\sum_{j = 1}^n {A_{ij}^2} })^{1/2}$ denote the Frobenius norm. For a given $m \times n$ matrix $A$ and an $n \times p$ matrix $B$, we want to get an approximation to the
product $AB$. In the following fast Monte Carlo algorithm in \cite{drineas2006fast}, we do $r$ independent trials. In each trial we randomly sample an element of $\{ 1,2, \cdots ,n\} $ with given discrete distribution $P=:\{ {p_i}\} _{i = 1}^n$. Then we extract an $m \times r$ matrix $C$ from the columns of $A$, and extract an $r \times n$ matrix $R$ from the corresponding rows of $B$. If the $P$ is chosen appropriately in the sense that $CR$ is a nice approximation to $AB$,
then the F-norm matrix concentration inequality in Lemma \ref{lem-fMC} holds with high probability.

\begin{lemma}\label{lem-fMC} (Theorem 2.1 in \cite{drineas2006sampling}) Let $A^{(i)}$ be the $i$-th row of $A \in {R^{m \times n}}$ as row vector and $B_{(j)}$ be the $j$-th column of $B \in {R^{n \times p}}$ as column vector. Suppose sampling probabilities $\{ {p_i}\} _{i = 1}^n,(\sum\nolimits_{i = 1}^n {{p_i} = 1} )$ are such that
\[{p_i} \ge \beta \frac{{\left\| {{A^{(i)}}} \right\|\left\| {{B_{(j)}}} \right\|}}{{\sum\nolimits_{j = 1}^n {\left\| {{A^{(i)}}} \right\|\left\| {{B_{(j)}}} \right\|} }}\]
for some $\beta  \in (0,1]$. Construct $C$ and $R$ with Algorithm 1 in \cite{drineas2006fast}, and assume that $\varepsilon  \in (0,1/3)$. Then, with probability at least $1 - \varepsilon $, we have
\[{\left\| {AB - CR} \right\|_F} \le \frac{{4\sqrt {\log (1/\varepsilon )} }}{{\beta \sqrt c }}{\left\| A \right\|_F}{\left\| B \right\|_F}.\]
\end{lemma}

Now we prove Theorems ~\ref{thm:r-size} and \ref{thm:r-size1} by applying the above Lemma \ref{lem-fMC}.

\begin{proof}
Note the fact that the maximum likelihood estimate $\hat{\bm \beta}_{\rm MLE}$ of the parameter vector $\bm\beta$ satisfy the following estimation equation
\begin{equation}\label{eq:equation1}
 \bm X^T[\bm y-\dot{\psi}(u(\bm X^T{\bm \beta}))]\dot{u}(\bm X^T{\bm \beta})=\bm 0,
\end{equation}
{\purple where $\dot{\psi}(u(\bm X^T{\bm \beta}))$ denotes the $n \times n$  diagonal matrix whose  $i$-th element in its diagonal is $\dot{\psi}(u(\bm x_i^T{\bm \beta}))$.}

 Without of loss of generality, we only show the case with probability $\bm\pi^{\rm mV}$, since the proof for $\bm\pi^{\rm mVc}$ is quite similar.
 Let $S$ be an $n\times r$ matrix whose $i$-th column is $1/\sqrt{r{\pi}^{\mathrm{mV}}_{j_i}}\bm e_{j_i}$, where $\bm e_{j_i}\in \mathbb{R}^n$ denotes the all-zeros vector except that its $j_i$-th entry is set to one. Here $j_i$ denotes the $j_i$-th data point chosen from the $i$-th independent random subsampling with probabilities $\bm\pi^{\mathrm{mV}}$.
 Then $\tilde{\bm\beta}$ satisfies the following equation
\begin{equation}\label{eq:equation2}
 \bm X^TSS^T[\bm y-\dot{\psi}(u(\bm X^T{\bm \beta}))]\dot{u}(\bm X^T{\bm \beta})=\bm 0.
\end{equation}
 Let $\|\cdot\|_F$ denote the Frobenius norm, we have
 \begin{align*}
  &\quad~ \sigma_{\min}(\bm X^TSS^T\dot{u}(\bm X^T\tilde{\bm \beta}))\|[\dot{\psi}(u(\bm X^T\tilde{\bm\beta}))-\dot{\psi}(u(\bm X^T\hat{\bm\beta}_{\rm MLE}))]\|\\
  &\le\|{\bm X}^TSS^T\dot{u}(\bm X^T\tilde{\bm \beta})[\dot{\psi}(u(\bm X^T\tilde{\bm\beta}))-\dot{\psi}(u(\bm X^T\hat{\bm\beta}_{\rm MLE}))]\|_F\\
  &\le\|{\bm X}^TSS^T\dot{u}(\bm X^T\tilde{\bm \beta})[\dot{\psi}(u(\bm X^T\tilde{\bm\beta}))-\bm y]\|_F\\
  &+\|\bm X^TSS^T\dot{u}(\bm X^T\tilde{\bm \beta})[\bm y-\dot{\psi}(u(\bm X^T\hat{\bm\beta}_{\rm MLE}))]\|_F\\
  &=\|{\bm X}^TSS^T\dot{u}(\bm X^T\tilde{\bm \beta})[\bm y-\dot{\psi}(u(\bm X^T\hat{\bm\beta}_{\rm MLE}))]\|_F~~~~[by~(\ref{eq:equation2})]\\
  &\le \|\bm X^T\dot{u}(\bm X^T\tilde{\bm \beta})[\bm y-\dot{\psi}(u(\bm X^T\hat{\bm\beta}_{\rm MLE}))]\|_F\\
  &+\left\|\bm X^T\dot{u}(\bm X^T\tilde{\bm \beta})[\bm y-\dot{\psi}(u(\bm X^T\hat{\bm\beta}_{\rm MLE}))]-{\bm X}^TSS^T\dot{u}(\bm X^T\tilde{\bm \beta})[\bm y-\dot{\psi}(u(\bm X^T\hat{\bm\beta}_{\rm MLE}))]\right\|_F\\
  &\le \|\bm X^T\dot{u}(\bm X^T\tilde{\bm \beta})[\bm y-\dot{\psi}(u(\bm X^T\hat{\bm\beta}_{\rm MLE}))]\|_F\\
  &+\frac{4\kappa(\mathcal{J}_X^{-1})\sqrt{\log(1/\epsilon)}}{\sqrt{r}}\|\bm X\|_F\|\dot{u}(\bm X^T\tilde{\bm \beta})[\bm y-\dot{\psi}(u(\bm X^T\hat{\bm\beta}_{\rm MLE}))]\|\\
  &\le \sigma_{\max}(\bm X)\sqrt{p}\|\dot{u}(\bm X^T\tilde{\bm \beta})[\bm y-\dot{\psi}(u(\bm X^T\hat{\bm\beta}_{\rm MLE}))]\|\\
  &+\frac{4\kappa(\mathcal{J}_X^{-1})\sqrt{\log(1/\epsilon)}}{\sqrt{r}}\sigma_{\max}(\bm X)\sqrt{p}\|\dot{u}(\bm X^T\tilde{\bm \beta})[\bm y-\dot{\psi}(u(\bm X^T\hat{\bm\beta}_{\rm MLE}))]\|\\
  &\le[1+\frac{4\kappa(\mathcal{J}_X^{-1})\sqrt{\log(1/\epsilon)}}{\sqrt{r}}]\sigma_{\max}(\bm X)\sqrt{p}\|\dot{u}(\bm X^T\tilde{\bm \beta})[\bm y-\dot{\psi}(u(\bm X^T\hat{\bm\beta}_{\rm MLE}))]\|\\
  &\le C_{\dot{u}}[1+\frac{4\kappa(\mathcal{J}_X^{-1})\sqrt{\log(1/\epsilon)}}{\sqrt{r}}]\sigma_{\max}(\bm X)\sqrt{p}\|[\bm y-\dot{\psi}(u(\bm X^T\hat{\bm\beta}_{\rm MLE}))]\|
\end{align*}
 where the fourth last inequality follows from Lemma \ref{lem-fMC} by putting $A=\bm X^T\dot{u}(\bm X^T\tilde{\bm \beta}), B=\dot{u}(\bm X^T\tilde{\bm \beta})[\bm y-\dot{\psi}(u(\bm X^T\hat{\bm\beta}_{\rm MLE}))], C={\bm X}^T{\bm S}, R={\bm S}^T\dot{u}(\bm X^T\tilde{\bm \beta})(\bm y-\dot{\psi}(u(\bm X^T\hat{\bm\beta}_{\rm MLE})))$ and $\beta=1/\kappa(\mathcal{J}_X^{-1})$, and last equality stems from (H.1) and Remark \ref{rm:analytic} with ${C_{\dot u}} = \mathop {\sup }\limits_{r \in K \subset \Theta } \left| {\dot u(r)} \right|$.

 Hence,
 \begin{align}\label{eq:mid-r-size}
 & \|\dot{\psi}(u(\bm X^T\hat{\bm\beta}_{\rm MLE}))-\dot{\psi}(u(\bm X^T\tilde{\bm\beta}))\|\nonumber\\
 \le& C_{\dot{u}}\frac{[1+\frac{4\kappa(\mathcal{J}_X^{-1})\sqrt{\log(1/\epsilon)}}{\sqrt{r}}]{\sqrt p {\sigma _{\max }}(\bm X)}}{\sigma_{\min}(\dot{u}(\bm X^T\tilde{\bm \beta})\bm X^TSS^T)}\|[\bm y-\dot{\psi}(u(\bm X^T\hat{\bm\beta}_{\rm MLE}))]\|.
 \end{align}

{\purple Then by following the facts that
 $$\sigma_{\min}(\dot{u}(\bm X^T\tilde{\bm \beta})\bm X^TSS^T\bm X\dot{u}(\bm X^T\tilde{\bm \beta}))\le \sigma_{\max}(\dot{u}(\bm X^T\tilde{\bm \beta})\bm X)\sigma_{\min}(\dot{u}(\bm X^T\tilde{\bm \beta})\bm X^TSS^T)$$
 and $\sigma^2_{\min}(\dot{u}(\bm X^T\tilde{\bm \beta})\bm{\tilde{ X}}^T)=\sigma_{\min}(\dot{u}(\bm X^T\tilde{\bm \beta})\bm X^TSS^T\bm X\dot{u}(\bm X^T\tilde{\bm \beta}))\ge 0.5\sigma^2_{\min}(\bm X)$, it holds that
 \begin{equation}\label{eq:Xss}
 \sigma_{\min}(\dot{u}^2(\bm X^T\tilde{\bm \beta})\bm X^TSS^T)\ge 0.5\sigma^2_{\min}(\dot{u}(\bm X^T\tilde{\bm \beta})\bm X)/\sigma_{\max}(\dot{u}(\bm X^T\tilde{\bm \beta})\bm X).
 \end{equation}
Combing the result \eqref{eq:Xss} with (\ref{eq:mid-r-size}), the desired result holds
{\normalsize \begin{align}
 & \|\dot{\psi}(u(\bm X^T\hat{\bm\beta}_{\rm MLE}))-\dot{\psi}(u(\bm X^T\tilde{\bm\beta}))\|\nonumber\\
 & \le 2C_{\dot{u}}{[1+\frac{4\alpha\sqrt{\log(1/\epsilon)}}{\sqrt{r}}]{\sqrt p \kappa^2(\dot{u}(\bm X^T\tilde{\bm \beta})\bm X)}}\|[\bm y-\dot{\psi}(u(\bm X^T\hat{\bm\beta}_{\rm MLE}))]\|.
 \end{align}}
}

Now, we turn to prove Theorem \ref{thm:r-size1}.

Note that
\[{p_i} \ge \alpha \frac{\min_{j}|y_j-\dot{\psi}(u(\hat{\bm\beta}_{\rm MLE}^T\bm x_j))||{\dot u}(\hat{\bm\beta}_{\rm MLE}^T\bm x_j)|}{\sqrt{\sum_j|y_j-\dot{\psi}(u(\hat{\bm\beta}_{\rm MLE}^T\bm x_j))|^2|{\dot u}^2(\hat{\bm\beta}_{\rm MLE}^T\bm x_j)|}}\frac{\|\bm x_i\|}{\sum_{j}\|\bm x_j\|}=\delta\frac{\|\bm x_i\|}{\sum_{j}\|\bm x_j\|},\]
with some $0<\delta:=\frac{\alpha\gamma}{\sqrt{\sum_j|y_j-\dot{\psi}(u(\hat{\bm\beta}_{\rm MLE}^T\bm x_j))|^2|{\dot u}^2(\hat{\bm\beta}_{\rm MLE}^T\bm x_j)|}}\le 1.$

{\purple According to the Weyl~inequality, we have
 \begin{align*}
&|\sigma_{\min}(\dot{u}(\bm X^T\tilde{\bm \beta})\bm X^T\bm \dot{u}(\bm X^T\tilde{\bm \beta})\bm X)-\sigma_{\min}(\dot{u}(\bm X^T\tilde{\bm \beta})\bm X^TSS^T\bm X\dot{u}(\bm X^T\tilde{\bm \beta}))|\\
\le& \|\dot{u}(\bm X^T\tilde{\bm \beta})\bm X^T\bm \dot{u}(\bm X^T\tilde{\bm \beta})X- \dot{u}(\bm X^T\tilde{\bm \beta})\bm X^TSS^T\bm X\dot{u}(\bm X^T\tilde{\bm \beta})\|_{S}\\
\le& \|\dot{u}(\bm X^T\tilde{\bm \beta})\|_S\|(\bm X^T\bm X-\bm X^TSS^T\bm X)\|_{S}\|\dot{u}(\bm X^T\tilde{\bm \beta})\|_S\\
 \le& {c_d} C_{\dot{u}}^2\|(\bm X^T\bm X-\bm X^TSS^T\bm X)\|_F\\
 \le& {c_d} \frac{4\sqrt{\log(1/\epsilon)}C_{\dot{u}}^2}{\delta\sqrt{r}}\|\bm X\|_F^2\\
 \le& {c_d} \frac{4\sqrt{\log(1/\epsilon)}C_{\dot{u}}^2}{\delta\sqrt{r}}p\sigma_{\max}^2(\bm X).
 \end{align*}

Using 
 the above inequality,  if we set
 \[r>64{c_d^2C_{\dot{u}}^2}\log(1/\epsilon)\sigma_{\max}^4(\bm X)p^2/(\delta^2\sigma_{\min}^4(\dot{u}(\bm X^T\tilde{\bm \beta})\bm X)),\]
 it holds that
  \begin{align*}
&|\sigma_{\min}(\dot{u}(\bm X^T\tilde{\bm \beta})\bm X^TSS^T\bm X\dot{u}(\bm X^T\tilde{\bm \beta}))-\sigma_{\min}(\dot{u}(\bm X^T\tilde{\bm \beta})\bm X^T\bm X\bm \dot{u}(\bm X^T\tilde{\bm \beta}))|\\
 \le& 0.5\sigma_{\min}(\dot{u}(\bm X^T\tilde{\bm \beta})\bm X^T\bm X\bm \dot{u}(\bm X^T\tilde{\bm \beta})).
 \end{align*}
 Thus the following equation holds with probability at least $1-\varepsilon$:
 \[\sigma_{\min}(\dot{u}(\bm X^T\tilde{\bm \beta})\bm X^TSS^T\bm X\dot{u}(\bm X^T\tilde{\bm \beta}))\ge 0.5\sigma^2_{\min}(\dot{u}(\bm X^T\tilde{\bm \beta})\bm X).\]
}
 \end{proof}
 \newpage

\subsection{Proof of {Theorem~\ref{thm:asy-2step-alg}}}

For the average weighted log-likelihood in Step 2 of two-step algorithm, we have
\begin{align*}
L_{{\tilde{\bm\beta}_0}}^{\rm{two-step}}(\bm\beta ): &= \frac{1}{{r + {r_0}}}\sum\limits_{i = 1}^{r + {r_0}} {\frac{{t_i^*(\bm\beta )}}{{\pi _i^*({\tilde{\bm\beta}_0})}}}  = \frac{1}{{r + {r_0}}}[\sum\limits_{i = 1}^{{r_0}} {\frac{{t_i^*(\bm\beta )}}{{\pi _i^*({\tilde{\bm\beta}_0})}} + \sum\limits_{i = {r_0} + 1}^{r + {r_0}} {\frac{{t_i^*(\bm\beta )}}{{\pi _i^*({\tilde{\bm\beta}_0})}}} } ]\\
&=\frac{{ {{r_0}} }}{{r + {r_0}}} \cdot \frac{1}{{ {{r_0}} }}\sum\limits_{i = 1}^{{r_0}} {\frac{{t_i^*(\bm\beta )}}{{\pi _i^*({\tilde{\bm\beta}_0})}} + \frac{{ r }}{{r + {r_0}}} \cdot \frac{1}{{ r }}\sum\limits_{i = {r_0} + 1}^{r + {r_0}} {\frac{{t_i^*(\bm\beta )}}{{\pi _i^*({\tilde{\bm\beta}_0})}}} },
\end{align*}
where $\pi _i^*({\tilde{\bm\beta}_0})$ in the first item stands for the initial subsampling strategy which satisfies (H.5).

For the sake of brevity, we begin with  the case with probability $\bm\pi^{\rm mV}$.
Denote the log-likelihood in the first and second steps by
\begin{equation*}
L^{*0}_{\tilde{\bm\beta}_0}(\bm\beta)=\frac{1}{{r_0}}\sum_{i=1}^{{r_0}} \frac{t_i^*(\bm\beta)}{\pi_i^*(\tilde{\bm\beta}_0)},\text{\quad and\quad}L^*_{\tilde{\bm\beta}_0}(\bm\beta)=\frac{1}{{r}}\sum_{i=1}^{{r}} \frac{t_i^*(\bm\beta)}{\pi_i^*(\tilde{\bm\beta}_0)},
\end{equation*}
respectively, {where $\pi_i(\tilde{\bm\beta}_0)=\check{\pi}_i^{\mathrm{opt}}$ in $L^*_{\tilde{\bm\beta}_0}(\bm\beta)$, and it has been calculated in the two-step algorithm in Section 4.} 

 To proof of Theorem~\ref{thm:asy-2step-alg}, we begin with the following Lemma \ref{lem:lem3}.

\begin{lemma}\label{lem:lem3}
If Assumptions~(H.1)--(H.4) holds, then as $n\rightarrow\infty$, conditionally on $\mathcal{F}_n$ in probability,
\begin{align}
  \check{\mathcal{J}}_X^{\tilde{\bm\beta}_0}-\mathcal{J}_X&=O_{P|\mathcal{F}_n}({r}^{-1/2}),\label{eq:4}\\
   \frac{1}{n}\frac{\partial L^*_{\tilde{\bm\beta}_0}(\hat{\bm\beta}_{\rm MLE})}{\partial\bm\beta} &=O_{P|\mathcal{F}_n}({r}^{-1/2}),\label{eq:28}
\end{align}
where {
\begin{eqnarray*}
\check{\mathcal{J}}_X^{\tilde{\bm\beta}_0}=
  -\frac{1}{n}\frac{\partial^2 L^*_{\tilde{\bm\beta}_0}(\hat{\bm\beta}_{\rm MLE})}
  {\partial\bm\beta\partial\bm\beta^T}
  &=&\frac{1}{nr}\sum_{i=1}^r\frac{\ddot{\psi}(u({{\bm\beta}^{T}{\bm x^*}_{i}}))
         \dot{u}({\bm\beta}^{T}{\bm x^*}_{i})\bm x_i^*[\dot{u}(\hat{\bm\beta}_{\rm MLE}^T{\bm x_{i}^*})\bm x_i^*]^T}{{\pi}_i^*(\tilde{\bm\beta}_0)}\\
         &&+\frac{1}{{nr}}\sum\limits_{i = 1}^r {\frac{{\ddot u({\hat{\bm\beta}_{\rm MLE}^T{\bm x_{i}^*}}){\bm x_i^*}{\bm x^*}_i^T[\dot \psi (u({\hat{\bm\beta}_{\rm MLE}^T{\bm x}_{i}^*})) - {y_i^*}]}}{{\pi _i^*(\tilde{\bm\beta}_0)}}}.
\end{eqnarray*}
}
\end{lemma}

\begin{proof}
Using the same arguments in Lemma \ref{lem:lem1}, we have {
\begin{equation}\label{eq:twoterms}
\begin{split}
&  {\rm E}\left(\check{\mathcal{J}}_X^{\tilde{\bm\beta}_0,j_1j_2}-\mathcal{J}_X^{j_1j_2}\Big|\mathcal{F}_n,\tilde{\bm\beta}_0\right)^2 \le \frac{{O_P(1)}}{r}\bigg[ \sum\limits_{i = 1}^n {\frac{{{{\dot u}^2}(\hat{\bm \beta} _{{\rm{MLE}}}^T{\bm x_i}){{({x_{i{j_1}}}{x_{i{j_2}}})}^2}}}{{{n^2}{\pi _i(\tilde{\bm\beta}_0)}}}}\\
&  +  \sum\limits_{i = 1}^n {\frac{{{{\{{{\dot u}^2}(\hat {\bm\beta} _{{\rm{MLE}}}^T{\bm x_i}){x_{i{j_1}}}{x_{i{j_2}}}[\dot \psi (u(\hat {\bm\beta} _{{\rm{MLE}}}^T{\bm x_i})) - {y_i}]\}}^2}}}{{{n^2}{\pi _i(\tilde{\bm\beta}_0)}}}}  \bigg].
  \end{split}
\end{equation}}
Now we substitute expression of ${\pi _i(\tilde{\bm\beta}_0)}$ in the two-step algorithm: $\tilde{\bm\pi}^{\mathrm{mV}}$ and $\tilde{\bm\pi}^{\mathrm{mVc}}$. Here we only give the proof of the case $\tilde{\bm\pi}^{\mathrm{mV}}$, and the proof of the case $\tilde{\bm\pi}^{\mathrm{mVc}}$ is analogous thus we omit it. For the first terms in \eqref{eq:twoterms}, note that ${\sigma _{\max }}(\tilde{\cal J}_X^{ - 1}),{\sigma _{\min }}(\tilde{\cal J}_X^{ - 1})$ are bounded from Lemma \ref{lem:lem1} and (H.4), it implies
\begin{align*}
& \sum\limits_{i = 1}^n {\frac{{{{\dot u}^2}(\hat {\bm\beta} _{{\rm{MLE}}}^T{\bm x_i}){{({x_{i{j_1}}}{x_{i{j_2}}})}^2}}}{{{n^2}{\pi _i(\tilde{\bm\beta}_0)}}}} \\
&\le \sum\limits_{i = 1}^n {\frac{{{{\left\| {{{\dot u}^2}(\hat {\bm\beta} _{{\rm{MLE}}}^T{\bm x_i}){\bm x_i}} \right\|}^2}\sum\limits_{j = 1}^n {\max } (|{y_j} - \dot \psi (u(\tilde{\bm \beta} _0^T{\bm x_j}))|,\delta )\left\| {{\cal J}_X^{ - 1}\dot u(\tilde{\bm\beta} _0^T{x_i}){x_i}} \right\|}}{{{n^2}\max (|{y_j} - \dot \psi (u(\tilde{\bm \beta} _0^T{\bm x_j}))|,\delta )\left\| {{\cal J}_X^{ - 1}\dot u(\tilde{\bm \beta} _0^T{\bm x_i}){\bm x_i}} \right\|}}} \\
&\le \sum\limits_{i = 1}^n {\frac{{{{\left\| {{{\dot u}^2}(\hat {\bm\beta} _{{\rm{MLE}}}^T{\bm x_i}){\bm x_i}} \right\|}^2}\sum\limits_{j = 1}^n {\max } (|{y_j} - \dot \psi (u(\tilde{\bm \beta} _0^T{\bm x_j}))|,\delta ){\sigma _{\max }}({\cal J}_X^{ - 1})\left\| {{{\dot u}^2}(\tilde{\bm \beta} _0^T{\bm x_i}){\bm x_i}} \right\|}}{{{n^2}\delta {\sigma _{\min }}({\cal J}_X^{ - 1})\left\| {{{\dot u}^2}(\tilde{\bm \beta} _0^T{\bm x_i}){\bm x_i}} \right\|}}} \\
&\le \kappa ({\cal J}_X^{ - 1})\sum\limits_{i = 1}^n {\frac{{\| {{{\dot u}^2}(\hat{\bm \beta} _{{\rm{MLE}}}^T{\bm x_i}){\bm x_i}} \|}}{{{n^2}\delta }}}\bigg[ \sum\limits_{j = 1}^n \frac{{|{y_j} - \dot \psi (u(\tilde{\bm \beta} _0^T{\bm x_j}))|\| {{{\dot u}^2}(\tilde {\bm\beta} _0^T{\bm x_i}){\bm x_j}} \|}}{n}\\
  &+ \sum\limits_{j = 1}^n {\frac{{\delta \| {{{\dot u}^2}(\tilde{\bm \beta} _0^T{\bm x_i}){\bm x_j}} \|}}{n}}   \bigg]\\
& \le \kappa ({\cal J}_X^{ - 1})\sum\limits_{i = 1}^n {\frac{{\| {{{\dot u}^2}(\hat{\bm \beta} _{{\rm{MLE}}}^T{x_i}){x_i}} \|}}{{n\delta }}} \bigg[ \sqrt {\sum\limits_{j = 1}^n {\frac{{|{y_j} - \dot \psi (u(\tilde{\bm \beta} _0^T{\bm x_j})){|^2}}}{n}} } \sqrt {\sum\limits_{j = 1}^n {\frac{{{{\| {{{\dot u}^2}(\tilde{\bm \beta} _0^T{\bm x_i}){\bm x_j}} \|}^2}}}{n}} }\\
&  + \sum\limits_{j = 1}^n {\frac{{\delta \| {{{\dot u}^2}( \tilde{\bm\beta} _0^T{\bm x_i}){\bm x_j}} \|}}{n}}  \bigg]\\
& \le O_P(1)\kappa ({\cal J}_X^{ - 1})\sum\limits_{i = 1}^n {\frac{{\| {{\bm x_i}} \|}}{{n\delta }}} \left[ {\sqrt {\sum\limits_{j = 1}^n {\frac{{|{y_j} - \dot \psi (u(\tilde{\bm \beta} _0^T{\bm x_j})){|^2}}}{n}} } \sqrt {\sum\limits_{j = 1}^n {\frac{{{{\| {{\bm x_j}} \|}^2}}}{n}} }  + \sum\limits_{j = 1}^n {\frac{{\delta \| {{\bm x_j}} \|}}{n}} } \right]\\
& = {O_P}(1).
\end{align*}
where the last equality is from (H.3) and (H.5).

For the second terms in \eqref{eq:twoterms},  we have
\begin{align*}
& \sum\limits_{i = 1}^n {\frac{{{{({{\dot u}^2}(\hat{\bm \beta} _{{\rm{MLE}}}^T{\bm x_i}){x_{i{j_1}}}{x_{i{j_2}}}[\dot \psi (u(\hat{\bm \beta} _{{\rm{MLE}}}^T{\bm x_i})) - {y_i}])}^2}}}{{{n^2}{\pi _i(\tilde{\bm\beta}_0)}}}} \\
& \le \sum\limits_{i = 1}^n {\frac{{{{\left\| {{{\dot u}^2}(\hat{\bm \beta} _{{\rm{MLE}}}^T{\bm x_i}){\bm x_i}} \right\|}^2}{{\left| {\dot \psi (u(\hat{\bm \beta} _{{\rm{MLE}}}^T{\bm x_i})) - {y_i}} \right|}^2}}}{{{n^2}\delta \left\| {{\cal J}_X^{ - 1}\dot u(\tilde{\bm \beta} _0^T{\bm x_i}){\bm x_i}} \right\|}}}\\
 &\times \sum\limits_{j = 1}^n {(|{y_j} - \dot \psi (u(\tilde{\bm \beta} _0^T{\bm x_j}))|{\rm{ + }}\delta )} \left\| {{\cal J}_X^{ - 1}\dot u(\tilde{\bm \beta} _0^T{\bm x_j}){\bm x_j}} \right\|\\
& \le \kappa ({\cal J}_X^{ - 1})\sum\limits_{i = 1}^n \bigg[{\frac{{\left\| {{{\dot u}^2}(\hat{\bm \beta} _{{\rm{MLE}}}^T{\bm x_i}){\bm x_i}} \right\|{{\left| {\dot \psi (u(\hat{\bm \beta} _{{\rm{MLE}}}^T{\bm x_i})) - {y_i}} \right|}^2}}}{{n\delta }}}\\
 &\times \frac{{\sum\limits_{j = 1}^n {(|{y_j} - \dot \psi (u(\tilde{\bm \beta} _0^T{\bm x_j}))|{\rm{ + }}\delta )} \left\| {\dot u(\tilde{\bm \beta} _0^T{\bm x_j}){\bm x_j}} \right\|}}{n}\bigg]\\
& = \kappa ({\cal J}_X^{ - 1})\sum\limits_{i = 1}^n {\frac{{\left\| {{{\dot u}^2}(\hat{\bm \beta} _{{\rm{MLE}}}^T{\bm x_i}){\bm x_i}} \right\|{{\left| {\dot \psi (u(\hat{\bm \beta} _{{\rm{MLE}}}^T{\bm x_i})) - {y_i}} \right|}^2}}}{{n\delta }}} {O_P}(1)\\
& \le \frac{{\kappa ({\cal J}_X^{ - 1})}}{\delta }\sqrt {\sum\limits_{i = 1}^n {\frac{{{{\left\| {{{\dot u}^2}(\hat{\bm \beta} _{{\rm{MLE}}}^T{\bm x_i}){\bm x_i}} \right\|}^2}}}{n}} } \sqrt {\sum\limits_{i = 1}^n {\frac{{{{\left| {\dot \psi (u(\hat{\bm \beta} _{{\rm{MLE}}}^T{\bm x_i})) - {y_i}} \right|}^4}}}{n}} } {O_P}(1)\\
& = {O_P}(1)
\end{align*}
where the second last and last equality is from (H.3) and (H.5).

Direct calculation yields
\begin{equation*}
E{\left( { \check{\mathcal{J}}_X^{\tilde{\bm\beta}_0,j_1j_2}- {\cal J}_X^{{j_1}{j_2}}|{{\cal F}_n}} \right)^2} = {E_{\tilde{\bm \beta} _0^{}}}E{\left( { \check{\mathcal{J}}_X^{\tilde{\bm\beta}_0,j_1j_2}- {\cal J}_X^{{j_1}{j_2}}|{{\cal F}_n},\tilde{\bm \beta} _0^{}} \right)^2}=O_P(r^{-1})
\end{equation*}
where $E_{\tilde{\bm\beta}_0}$ means that the expectation is taken with respect to the distribution of $\tilde{\bm\beta}_0$ given $\mathcal{F}_n$. 

On the other hand, following the same arguments in Lemma \ref{lem:lem1}, we can have
\[
E\left\{\frac{L^*_{\tilde{\bm\beta}_0}(\bm\beta)}{n}-\frac{L(\bm\beta)}{n}\bigg|\mathcal{F}_n,\tilde{\bm\beta}_0\right\}^2
  =O_P(r^{-1}).
\]

Then $E\left\{n^{-1}{L^*_{\tilde{\bm\beta}_0}(\bm\beta)}-n^{-1}{L(\bm\beta)}|\mathcal{F}_n\right\}^2=O_P(r^{-1})$.

Similarly, we can see that $\text{Var}(n^{-1}{\partial L^*_{\tilde{\bm\beta}_0}(\hat{\bm\beta}_{\rm MLE})}/{\partial\bm\beta})=O_P(r^{-1})$.
 Thus, the desired result holds.
\end{proof}

Now we prove Theorem~\ref{thm:asy-2step-alg}.
\begin{proof}
Using the same arguments in Lemma \ref{lem:lem3} 
, we have
{
\[\begin{split}
&E\left\{\frac{L^{\rm{two-step}}_{\tilde{\bm\beta}_0}(\bm\beta)}{n}-\frac{L(\bm\beta)}{n}\bigg|\mathcal{F}_n\right\}^2\le 2(\frac{{ {{r_0}} }}{{r + {r_0}}})^2E\left\{\frac{L^{*0}_{\tilde{\bm\beta}_0}(\bm\beta)}{n}-\frac{L(\bm\beta)}{n}\bigg|\mathcal{F}_n\right\}^2\\
&+2(\frac{{ {{r}} }}{{r + {r_0}}})^2E\left\{\frac{L^{*}_{\tilde{\bm\beta}_0}(\bm\beta)}{n}-\frac{L(\bm\beta)}{n}\bigg|\mathcal{F}_n\right\}^2
  =O_P(r^{-1}).
  \end{split}
\]
}
Therefore $E\{n^{-1}{L^{\rm{two-step}}_{\tilde{\bm\beta}_0}(\bm\beta)}-n^{-1}{L(\bm\beta)}|\mathcal{F}_n\}^2\rightarrow 0$ as $r_0/r\rightarrow 0,r\rightarrow\infty$ and $n^{-1}{L^{\rm{two-step}}_{\tilde{\bm\beta}_0}(\bm\beta)}-n^{-1}{L(\bm\beta)}\rightarrow 0$ in conditional probability given $\mathcal{F}_n$. Also note that the parameter space is compact and $\hat{\bm\beta}_{\rm MLE}$ is the unique global maximum of the continuous convex function $L(\bm\beta)$. Thus, from Theorem 5.9 and its remark of \cite{Vaart2000Asymptotic}, we have
\[
  \|{\breve{\bm\beta}}-\hat{\bm\beta}_{\rm MLE}\|=o_{P|\mathcal{F}_n}(1).
\]
 Using Taylor's theorem,
 \begin{align*}
0&=\frac{\dot{L}^{\rm{two-step}}_{\tilde{\bm\beta_0},j}(\breve{\bm\beta})}{n}=\frac{{ {{r_0}} }}{{r + {r_0}}}\frac{\dot{L}^{*0}_{\tilde{\bm\beta_0},j}(\breve{\bm\beta})}{n}+\frac{{ {{r}} }}{{r + {r_0}}}\frac{\dot{L}^{*}_{\tilde{\bm\beta_0},j}(\breve{\bm\beta})}{n}\\
&=\frac{{ {{r}} }}{{r + {r_0}}}\left\{\frac{\dot{L}^*_{\tilde{\bm\beta_0},j}(\hat{\bm\beta}_{\rm MLE})}{n}
   +\frac{1}{n} \frac{\partial\dot{L}^*_{\tilde{\bm\beta_0},j}(\hat{\bm\beta}_{\rm MLE})}{\partial \bm\beta^T}(\breve{\bm\beta}-\hat{\bm\beta}_{\rm MLE})
   +\frac{1}{n} R_{\tilde{\bm\beta_0},j}\right\}\\
   &+\frac{{ {{r_0}} }}{{r + {r_0}}}\frac{\dot{L}^{*0}_{\tilde{\bm\beta_0},j}(\breve{\bm\beta})}{n},
\end{align*}
where $\dot{L}^*_{\tilde{\bm\beta_0},j}({\bm\beta})$ is the partial derivative of $L^*_{\tilde{\bm\beta_0},j}({\bm\beta})$ with respect to $\beta_j$.

By similar argument in the Proof of Theorem~\ref{thm:as-general-alg}, the Lagrange remainder have the rate
{\normalsize \begin{align*}
  \frac{1}{n}R_{\tilde{\bm\beta_0},j}
  &:=\frac{1}{n}(\breve{\bm\beta}-\hat{\bm\beta}_{\rm MLE})^T
   \int_0^1\int_0^1\frac{\partial^2\dot{L}^*_j\{\hat{\bm\beta}_{\rm MLE}
   +uv(\breve{\bm\beta}-\hat{\bm\beta}_{\rm MLE})\}}{\partial \bm\beta\partial
   \bm\beta^T}v du dv\
(\breve{\bm\beta}-\hat{\bm\beta}_{\rm MLE})\\
&=O_{P|\mathcal{F}_n}(\|\breve{\bm\beta}-\hat{\bm\beta}_{\rm MLE}\|^2).
\end{align*}}
Note that the subsampling probabilities in the first stage satisfies the condition (H.1)-(H.7), thus from Theorem \ref{thm:CLT}, it holds that
\begin{align*}
\frac{\dot{L}^{*0}_{\tilde{\bm\beta_0},j}(\breve{\bm\beta})}{n}
=&\frac{\dot{L}^{*0}_{\tilde{\bm\beta_0},j}(\hat{\bm\beta}_{\rm MLE})}{n}
   +\frac{1}{n} \frac{\partial\dot{L}^{*0}_{\tilde{\bm\beta_0},j}(\hat{\bm\beta}_{\rm MLE})}{\partial \bm\beta^T}(\breve{\bm\beta}-\hat{\bm\beta}_{\rm MLE})
   +O_{P|\mathcal{F}_n}(\|\breve{\bm\beta}-\hat{\bm\beta}_{\rm MLE}\|^2).
\end{align*}
Therefore
\begin{align*}
\frac{1}{n}\frac{\partial{L}^{*0}_{\tilde{\bm\beta_0}}(\breve{\bm\beta})}{\partial\bm\beta}
=&\frac{1}{n}\frac{\partial{L}^{*0}_{\tilde{\bm\beta_0}}(\hat{\bm\beta}_{\rm MLE})}{\partial\bm\beta}
   +\frac{1}{n} \frac{\partial^2{L}^{*0}_{\tilde{\bm\beta_0}}(\hat{\bm\beta}_{\rm MLE})}{\partial\bm\beta\partial \bm\beta^T}(\breve{\bm\beta}-\hat{\bm\beta}_{\rm MLE})
   +O_{P|\mathcal{F}_n}(\|\breve{\bm\beta}-\hat{\bm\beta}_{\rm MLE}\|^2).
\end{align*}
From Lemma \ref{lem:lem1}, it is clear to see that
\[ \frac{1}{n}\frac{\partial L^{*0}_{\tilde{\bm\beta_0}}(\hat{\bm\beta}_{\rm MLE})}{\partial\bm\beta}
   =O_{P|\mathcal{F}_n}(r_0^{-1/2})\]
   for the first step, since $\pi_i^*$ is prespecified and satisfied (H.6), and 
\[\frac{r_0}{r+r_0}\frac{1}{n}\frac{\partial L^{*0}_{\tilde{\bm\beta_0}}(\hat{\bm\beta}_{\rm MLE})}{\partial\bm\beta}
=\frac{r_0}{r}O_{P|\mathcal{F}_n}(r_0^{-1/2})=o_{P|\mathcal{F}_n}(r^{-1/2}),\]
since $r_0/r\to 0$. {
  This step holds due to the fact that $\frac{\sqrt{r_0}}{r}O_{P|\mathcal{F}_n}(1)=\frac{\sqrt{r_0}}{\sqrt{r}}O_{P|\mathcal{F}_n}(1)O_{P|\mathcal{F}_n}(r^{-1/2})=o(1)O_{P|\mathcal{F}_n}(r^{-1/2})$.}
Let
\[\breve{\mathcal{J}}_X:=\frac{r}{r+r_0}\frac{1}{n} \frac{\partial^2{L}^{*}_{\tilde{\bm\beta_0}}(\hat{\bm\beta}_{\rm MLE})}{\partial\bm\beta\partial \bm\beta^T}+\frac{r_0}{r+r_0}\frac{1}{n} \frac{\partial^2{L}^{*0}_{\tilde{\bm\beta_0}}(\hat{\bm\beta}_{\rm MLE})}{\partial\bm\beta\partial \bm\beta^T}.\]
Combine Lemmas \ref{lem:lem1} and \ref{lem:lem3}, we have
\[\begin{split}
\breve{\mathcal{J}}_X-\mathcal{J}_X
=&\frac{r}{r+r_0}\left(\check{\mathcal{J}}_X^{\tilde{\bm\beta}_0}-\mathcal{J}_X\right)+\frac{r_0}{r+r_0}\left(\frac{1}{n} \frac{\partial^2{L}^{*0}_{\tilde{\bm\beta_0}}(\hat{\bm\beta}_{\rm MLE})}{\partial\bm\beta\partial \bm\beta^T}-\mathcal{J}_X\right)\\
=&\frac{r}{r+r_0}O_{P|\mathcal{F}_n}(r^{-1/2})+\frac{r_0}{r+r_0}O_{P|\mathcal{F}_n}(r_0^{-1/2})=O_{P|\mathcal{F}_n}(r^{-1/2}),
\end{split}\]
since $r_0/r\rightarrow 0$.

Hence,
{\begin{align*}\label{eq:13}
  \breve{\bm\beta}-\hat{\bm\beta}_{\rm MLE}
  =& -(\breve{\mathcal{J}}_X)^{-1}\left\{\frac{1}{n}\dot {L}^*_{\tilde{\bm\beta}_0}(\hat{\bm\beta}_{\rm MLE})
    +O_{P|\mathcal{F}_n}(\|\breve{\bm\beta}-\hat{\bm\beta}_{\rm MLE}\|^2)+o_{P|\mathcal{F}_n}(r^{-1/2})\right\},\\
  =&O_{P|\mathcal{F}_n}(r^{-1/2})+o_{P|\mathcal{F}_n}(\|\breve{\bm\beta}-\hat{\bm\beta}_{\rm MLE}\|)
\end{align*}}
as $r_0/r\rightarrow0$, by noting $(\breve{\mathcal{J}}_X)^{-1} =O_{P|\mathcal{F}_n}(1)$ from (H.5).
Therefore, the desired result follows by noting
\[  \breve{\bm\beta}-\hat{\bm\beta}_{\rm MLE}=O_{P|\mathcal{F}_n}(r^{-1/2}).\]
\end{proof}

\subsection{Proof of {Theorem~\ref{thm:CLT-2step}}}
\begin{proof}
For the sake of brevity, we begin with  the case with probability $\tilde{\bm\pi}^{\rm mVc}$.
Denote
\begin{equation}\label{eq:56}
  \frac{\dot L^*_{\tilde{\bm\beta}_0}(\hat{\bm\beta}_{\rm MLE})}{n}
  =\frac{1}{{r}}\sum_{i=1}^{{r}}
  \frac{\{y^*_i-\dot{\psi}(u(\hat{\bm\beta}_{\rm MLE}^T\bm x_i^*))\}\dot{u}(\hat{\bm\beta}_{\rm MLE}^T\bm x_i^*)\bm x^*_i}
  {n{\pi}^*_i(\tilde{\bm\beta}_0)}
  =:\frac{1}{{r}}\sum_{i=1}^{{r}}\eta_i^{\tilde{\bm\beta}_0}.
\end{equation}
It can be shown that given $\mathcal{F}_n$ and ${\tilde{\bm\beta}_0}$, $\eta_1^{\tilde{\bm\beta}_0}, \ldots, \eta_{{r}}^{\tilde{\bm\beta}_0}$ are i.i.d
random variables with zero mean and variance
\begin{align*}
  &\text{var}(\eta_i^{\tilde{\bm\beta}_0}|\mathcal{F}_n,\tilde{\bm\beta}_0)
    ={r}V_c^{\tilde{\bm\beta}_0}
    =\frac{1}{n^2}\sum_{i=1}^n \pi_i(\tilde{\bm\beta}_0)\frac{\{y^*_i-\dot{\psi}(u(\hat{\bm\beta}_{\rm MLE}^T\bm x_i^*))\}^2\dot{u}^2(\hat{\bm\beta}_{\rm MLE}^T\bm x_i^*)
    \bm x_i\bm x_i^T}{{\pi}^2_i(\tilde{\bm\beta}_0)}.
\end{align*}
Meanwhile, for every $\varepsilon>0$,
\begin{align*}
   &\sum_{i=1}^{{r}} E\{\|{r}^{-1/2}\eta_i^{\tilde{\bm\beta}_0}\|^2
    I(\|\eta_i^{\tilde{\bm\beta}_0}\|>{r}^{1/2}\varepsilon)
    |\mathcal{F}_n,\tilde{\bm\beta}_0\}\\
   \le &\frac{1}{{r}^{3/2}\varepsilon}
    \sum_{i=1}^{{r}} E\{\|\eta_i^{\tilde{\bm\beta}_0}\|^{3}
    I(\|\eta_i^{\tilde{\bm\beta}_0}\|>{r}^{1/2}\varepsilon)|
    \mathcal{F}_n,\tilde{\bm\beta}_0\} \\
  \le &\frac{1}{{r}^{3/2}\varepsilon}
    \sum_{i=1}^{{r}} E(\|\eta_i^{\tilde{\bm\beta}_0}\|^{3}|
    \mathcal{F}_n,\tilde{\bm\beta}_0)\\
 \le &\frac{1}{{r}^{1/2}}\frac{1}{n^{3}}
    \sum_{i=1}^n\frac{\{|y_i-\dot{\psi}(u(\hat{\bm\beta}_{\rm MLE}^T\bm x_i))|\}^{3}
    \|\dot{u}(\hat{\bm\beta}_{\rm MLE}^T\bm x_i)\bm x_i\|^{3}}{\pi_i^{2}(\tilde{\bm\beta}_0)}\\
  \le&\frac{1}{{r}^{1/2}}\frac{1}{n}
    \sum_{i=1}^n\frac{\{|y_i-\dot{\psi}(u(\hat{\bm\beta}_{\rm MLE}^T\bm x_i))|\}^{2}
    \|\dot{u}(\hat{\bm\beta}_{\rm MLE}^T\bm x_i)\bm x_i\|}{\delta}\\
    &\times \left(\frac{1}{n}\sum_{j=1}^n\max(|y_j-\dot{\psi}(u(\tilde{\bm\beta}_0^T\bm x_j))|,\delta)\|\dot{u}(\tilde{\bm\beta}_0^T\bm x_j)\bm x_j\|\right)^2\\
  \le&\frac{1}{{r}^{1/2}}\frac{1}{n}
    \sum_{i=1}^n\frac{\{|y_i-\dot{\psi}(u(\hat{\bm\beta}_{\rm MLE}^T\bm x_i))|\}^{2}
    \|\dot{u}(\hat{\bm\beta}_{\rm MLE}^T\bm x_i)\bm x_i\|}{\delta}\\
    &\times \left(\frac{1}{n}\sum_{j=1}^n(|y_j-\dot{\psi}(u(\tilde{\bm\beta}_0^T\bm x_j))|+\delta)\|\dot{u}(\tilde{\bm\beta}_0^T\bm x_j)\bm x_j\|\right)^2.
\end{align*}

From (H.1), (H.3) and (H.5),
\begin{align*}
&\frac{1}{n}
    \sum_{i=1}^n\frac{\{|y_i-\dot{\psi}(u(\hat{\bm\beta}_{\rm MLE}^T\bm x_i))|\}^{2}
    \|\dot{u}(\hat{\bm\beta}_{\rm MLE}^T\bm x_i)\bm x_i\|}{\delta}\\
\le& \delta^{-1}\left(\frac{1}{n}
    \sum_{i=1}^n\{|y_i-\dot{\psi}(u(\hat{\bm\beta}_{\rm MLE}^T\bm x_i))|\}^{4}\right)^{1/2}\left(\frac{1}{n}
    \sum_{i=1}^n\|\dot{u}(\hat{\bm\beta}_{\rm MLE}^T\bm x_i)\bm x_i\|^2\right)^{1/2}\\
    =&O_P(1),
\end{align*}
by Holder's inequality.

Similarly, it can be shown
\[\frac{1}{n}\sum_{j=1}^n(|y_j-\dot{\psi}(u(\tilde{\bm\beta}_0^T\bm x_j))|+\delta)\|\dot{u}(\tilde{\bm\beta}_0^T\bm x_j)\bm x_j\|=O_P(1),\]
from (H.1), (H.3) and (H.5).

Hence
\[\sum_{i=1}^{{r}} E\{\|{r}^{-1/2}\eta_i^{\tilde{\bm\beta}_0}\|^2
    I(\|\eta_i^{\tilde{\bm\beta}_0}\|>{r}^{1/2}\varepsilon)
    |\mathcal{F}_n,\tilde{\bm\beta}_0\}=o_{P|\mathcal{F}_n}(1).\]
This shows that the Lindeberg-Feller conditions are satisfied in probability.
By the Lindeberg-Feller central limit theorem \citep[Proposition 2.27 of][]{Vaart2000Asymptotic}, conditionally on $\mathcal{F}_n$ and $\tilde{\bm\beta}_0$,
\begin{equation*}
  \frac{1}{n}(V_c^{\tilde{\bm\beta}_0})^{-1/2}\dot L^*(\hat{\bm\beta}_{\rm MLE})=
  \frac{1}{{r}^{1/2}}\{\text{var}(\eta_i|\mathcal{F}_n,\tilde{\bm\beta}_0)\}^{-1/2}\sum_{i=1}^{{r}}\eta_i
  \rightarrow N(0,I),
\end{equation*}
in distribution.

The distance between $V_c^{\tilde{\bm\beta}_0}$ and $V_c$ is
{\normalsize{
\begin{align*}
&~~~~\|V_c-V_c^{\tilde{\bm\beta}_0}\|\\&\le \frac{1}{r}\sum_{i=1}^n\left\|\frac{1}{\pi^{\rm mVc}_i}-\frac{1}{\pi_i(\tilde{\bm\beta}_0)}\right\|\frac{\{y_i-\dot{\psi}(u(\hat{\bm\beta}_{\rm MLE}^T\bm x_i))\}^2\dot{u}^2(\hat{\bm\beta}_{\rm MLE}^T\bm x_i)\|\bm x_i\|^2}{n}\\
&= \frac{1}{r}\sum_{i=1}^n\left\|1-\frac{\pi^{\rm mVc}_i}{\pi_i(\tilde{\bm\beta}_0)}\right\|\frac{\{y_i-\dot{\psi}(u(\hat{\bm\beta}_{\rm MLE}^T\bm x_i))\}^2\dot{u}^2(\hat{\bm\beta}_{\rm MLE}^T\bm x_i)\|\bm x_i\|^2}{n\pi^{\rm mVc}_i}\\
&\le \frac{1}{r}\sum_{i=1}^n\left\|1-\frac{\pi^{\rm mVc}_i}{\pi_i(\tilde{\bm\beta}_0)}\right\|\frac{|y_i-\dot{\psi}(u(\hat{\bm\beta}_{\rm MLE}^T\bm x_i))|\|\dot{u}(\hat{\bm\beta}_{\rm MLE}^T\bm x_i)\bm x_i\|}{n}\\
&\le\frac{1}{r}\left(\frac{1}{n}\sum_{i=1}^n\left\|1-\frac{\pi^{\rm mVc}_i}{\pi_i(\tilde{\bm\beta}_0)}\right\|^2\right)^{1/2}\left(\sum_{i=1}^n\frac{\{y_i-\dot{\psi}(u(\hat{\bm\beta}_{\rm MLE}^T\bm x_i))\}^2\dot{u}^2(\hat{\bm\beta}_{\rm MLE}^T\bm x_i)\|\bm x_i\|^2}{n}\right)^{1/2}\\
 & =o_{P|\mathcal{F}_n}({r}^{-1}),
  \end{align*}}}
  where the last equation follows from  the facts that
  \[ \left\|1-\frac{\pi^{\rm mVc}_i}{\pi_i(\tilde{\bm\beta}_0)}\right\|^2\le\frac{(\pi^{\rm mVc}_i-\pi_i(\tilde{\bm\beta}_0))^2}{\delta^2}=o_P(1),\]
  and
  \[\sum_{i=1}^n\frac{\{y_i-\dot{\psi}(u(\hat{\bm\beta}_{\rm MLE}^T\bm x_i))\}^2\dot{u}^2(\hat{\bm\beta}_{\rm MLE}^T\bm x_i)\|\bm x_i\|^2}{n}=O_P(1).\]
Here the first equality in the fact above holds for the continues mapping theorem and the second equality holds from (H.1), (H.3), (H.5) and Cauchy's inequality.

 Utilizing the facts
\begin{align*}
  (\check{\mathcal{J}}_X^{\tilde{\bm\beta}_0})^{-1}-\breve{\mathcal{J}}_X^{-1}
  &=-\breve{\mathcal{J}}_X^{-1}(\check{\mathcal{J}}_X^{\tilde{\bm\beta}_0}-\mathcal{J}_X+\mathcal{J}_X-\breve{\mathcal{J}}_X)
    (\check{\mathcal{J}}_X^{\tilde{\bm\beta}_0})^{-1}
  =O_{P|\mathcal{F}_n}({r}^{-1/2}),
\end{align*}
we have $(\check{\mathcal{J}}_X^{\tilde{\bm\beta}_0})^{-1}-(\breve{\mathcal{J}}_X)^{-1}=O_{P|\mathcal{F}_n}(r^{-1/2})$ from Lemma~\ref{lem:lem3} and~Theorem \ref{thm:asy-2step-alg}.
Thus
\begin{equation}\label{eq:57}
  {\breve{\bm\beta}}-\hat{\bm\beta}_{\rm MLE}{ =}
  -\frac{1}{n}(\check{\mathcal{J}}_X^{\tilde{\bm\beta}_0})^{-1}
  \dot L^*_{\tilde{\bm\beta}_0}(\hat{\bm\beta}_{\rm MLE})+O_{P|\mathcal{F}_n}({r}^{-1})
\end{equation}
Based on Equation (\ref{eq:4}), we further have
\begin{align*}
  (\check{\mathcal{J}}_X^{\tilde{\bm\beta}_0})^{-1}-\mathcal{J}_X^{-1}
  &=-\mathcal{J}_X^{-1}(\check{\mathcal{J}}_X^{\tilde{\bm\beta}_0}-\mathcal{J}_X)
    (\check{\mathcal{J}}_X^{\tilde{\bm\beta}_0})^{-1}
  =O_{P|\mathcal{F}_n}({r}^{-1/2}).
\end{align*}
Therefore
{{\begin{align*}
  &V^{-1/2}({\breve{\bm\beta}}-\hat{\bm\beta}_{\rm MLE})\\
  =&-V^{-1/2}{1\over n}(\check{\mathcal{J}}_X^{\tilde{\bm\beta}_0})^{-1}
    \dot L^*(\hat{\bm\beta}_{\rm MLE})
    +O_{P|\mathcal{F}_n}({r}^{-1/2})\\
  =&-V^{-1/2}\mathcal{J}_X^{-1}{1\over n}\dot L^*(\hat{\bm\beta}_{\rm MLE})
    -V^{-1/2}\{(\check{\mathcal{J}}_X^{\tilde{\bm\beta}_0})^{-1}-\mathcal{J}_X^{-1}\}
    {1\over n}\dot L^*(\hat{\bm\beta}_{\rm MLE})
    +O_{P|\mathcal{F}_n}({r}^{-1/2})\\
  =&-V^{-1/2}\mathcal{J}_X^{-1}(V_c^{\tilde{\bm\beta}_0})^{1/2}
    (V_c^{\tilde{\bm\beta}_0})^{-1/2}{1\over n}\dot L^*(\hat{\bm\beta}_{\rm MLE})
    +O_{P|\mathcal{F}_n}({r}^{-1/2}).
\end{align*}}}
 It can be shown that
 \begin{align*}
 &V^{-1/2}\mathcal{J}_X^{-1}(V_c^{\tilde{\bm\beta}_0})^{1/2}(V^{-1/2}\mathcal{J}_X^{-1}(V_c^{\tilde{\bm\beta}_0})^{1/2})^T\\
=&V^{-1/2}\mathcal{J}_X^{-1}(V_c^{\tilde{\bm\beta}_0})\mathcal{J}_X^{-1}V^{-1/2}\\
=&V^{-1/2}\mathcal{J}_X^{-1}(V_c)\mathcal{J}_X^{-1}V^{-1/2}+o_{P|\mathcal{F}_n}(r^{-1/2}) \\
=&I+o_{P|\mathcal{F}_n}(r^{-1/2}).
 \end{align*}
The desired result follows by Slutsky's theorem.

 As for the case $\pi_i(\tilde{\bm\beta}_0)=\tilde{\pi}_i^{\mathrm{mV}}$ in $L^*_{\tilde{\bm\beta}_0}(\bm\beta)$,  $\tilde{\pi}_i^{\mathrm{mV}}$ has the same expression as $\pi_i^{\mathrm{mV}}$ except that $\hat{\bm\beta}_{\rm MLE}$,  is replaced by $\tilde{\bm\beta}_0$. 
 Also note that
$\pi_i(\tilde{\bm\beta}_0)\ge \kappa({\tilde{\mathcal{J}}}_X)^{-1}\tilde{\pi}_i^{\mathrm{mVc}}$. The rest of the proof is the same as that of $\tilde{\pi}_i^{\mathrm{mVc}}$ with  minor modifications.
\end{proof}

\section{Additional Simulation Results}

In terms of the allocation between $r_0$ and $r$, it is clear to see  that the
two-step approach works the best when $r_0/r$ is around 0.2 from the simulation result in Figure \ref{fig:4} of the main text. To well demonstrate our methods, we also  compare different $r_0+ r$ with fixed $r_0/r = 0.2$.

In each of the settings described in Section \ref{sec:Poisson-sim} of the article, we reevaluated the performance of $\tilde{\pi}_i^{\mathrm{mV}}$ and $\tilde{\pi}_i^{\mathrm{mVc}}$ when $r_0/r$ is fixed at 0.2.
For comparison, the  uniform subsampling, leverage subsampling and adjust leverage subsampling methods are also considered.
Inline with the setting in the main text,  the sample size $r_0+r$ is selected as 500, 700, 900, 1200, 1400 and 1600, respectively.
We report the results for the Poisson regression and negative binomial regression in Figures \ref{fig:fixratio-poisson} and  \ref{fig:fixratio-nb}, respectively.

\begin{figure}[!htp]
  \centering
  \begin{subfigure}{0.49\textwidth}
    \includegraphics[width=\textwidth]{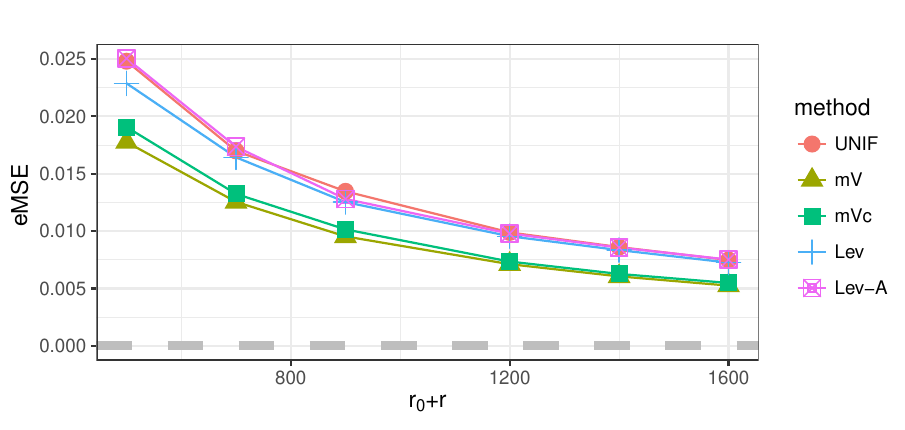}\\[-1cm]
    \caption{Case 1 }
  \end{subfigure}
  \begin{subfigure}{0.49\textwidth}
    \includegraphics[width=\textwidth]{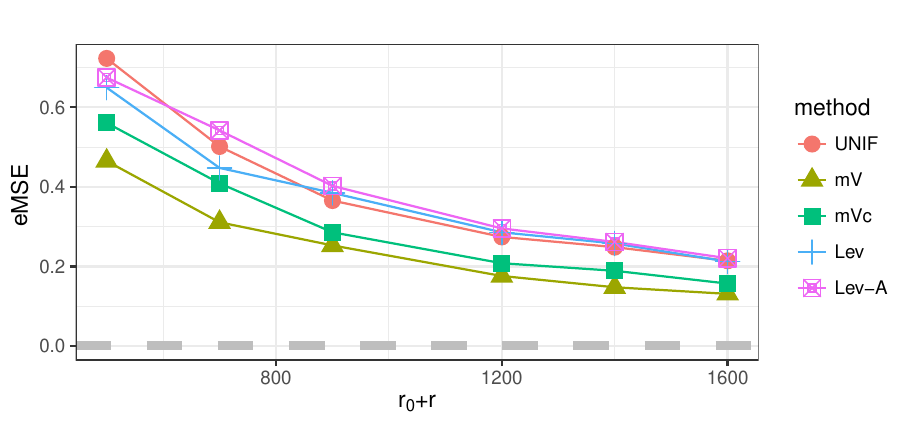}\\[-1cm]
    \caption{Case 2}
  \end{subfigure}\\[5mm]
  \begin{subfigure}{0.49\textwidth}
    \includegraphics[width=\textwidth]{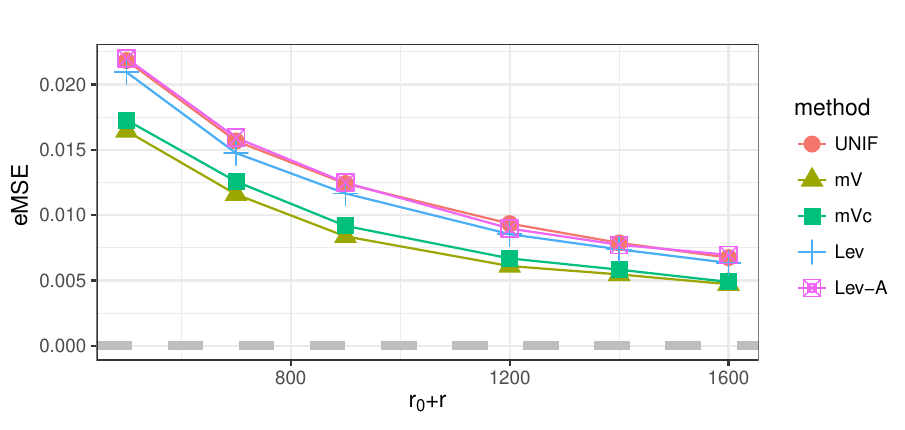}\\[-1cm]
    \caption{Case 3}
  \end{subfigure}
  \begin{subfigure}{0.49\textwidth}
    \includegraphics[width=\textwidth]{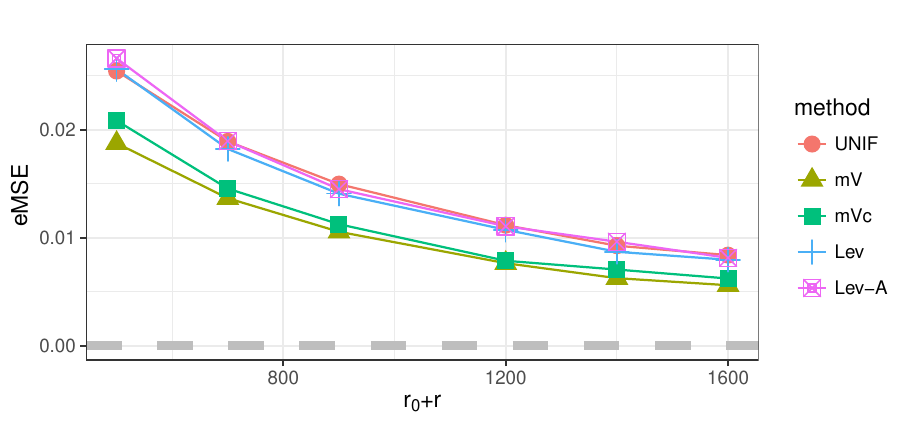}\\[-1cm]
    \caption{Case 4}
  \end{subfigure}\\[5mm]
  \caption{The eMSEs for Poisson regression with different  subsample size ${r_0+r}$ and a fixed  $r_0/r=0.2$.
     The different distributions of covariates are listed in the beginning of Section 5.}
  \label{fig:fixratio-poisson}
\end{figure}

\begin{figure}[!htp]
  \centering
  \begin{subfigure}{0.49\textwidth}
    \includegraphics[width=\textwidth]{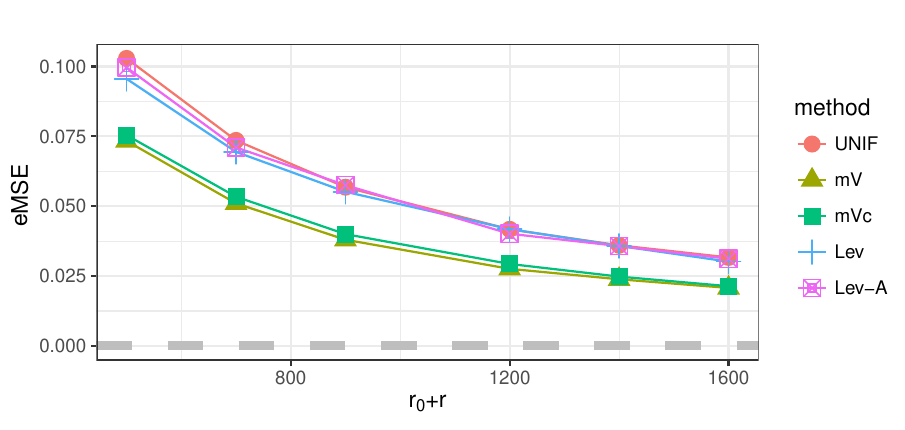}\\[-1cm]
    \caption{Case 1}
  \end{subfigure}
  \begin{subfigure}{0.49\textwidth}
    \includegraphics[width=\textwidth]{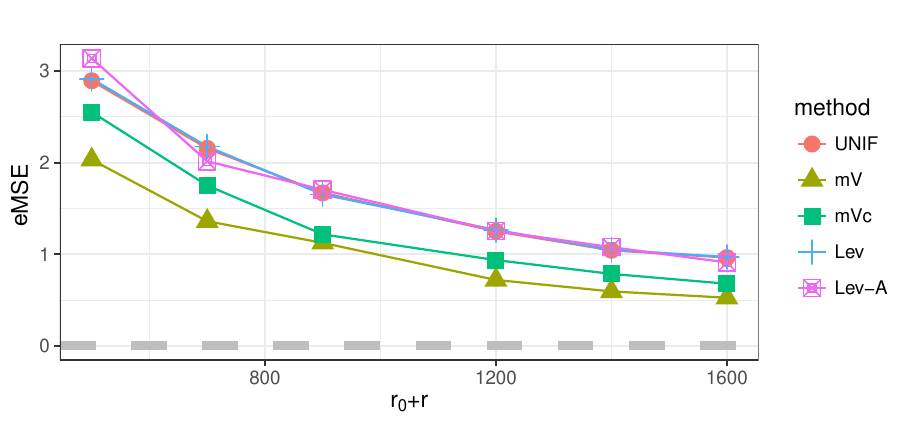}\\[-1cm]
    \caption{Case 2}
  \end{subfigure}\\[5mm]
  \begin{subfigure}{0.49\textwidth}
    \includegraphics[width=\textwidth]{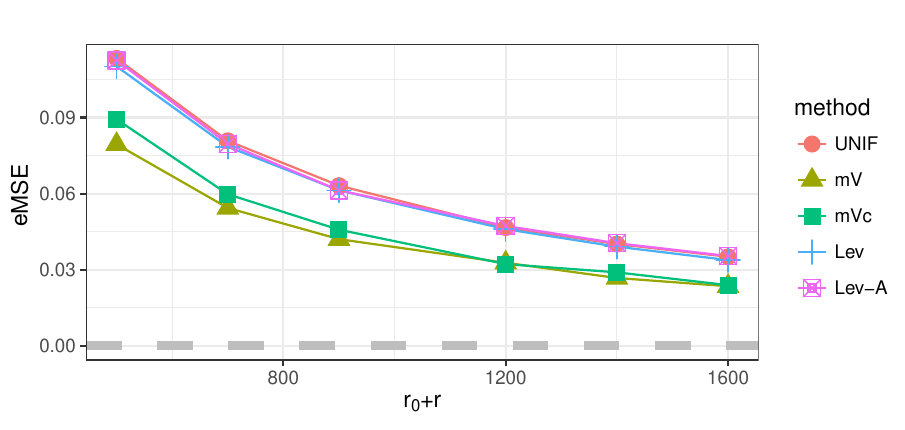}\\[-1cm]
    \caption{Case 3}
  \end{subfigure}
  \begin{subfigure}{0.49\textwidth}
    \includegraphics[width=\textwidth]{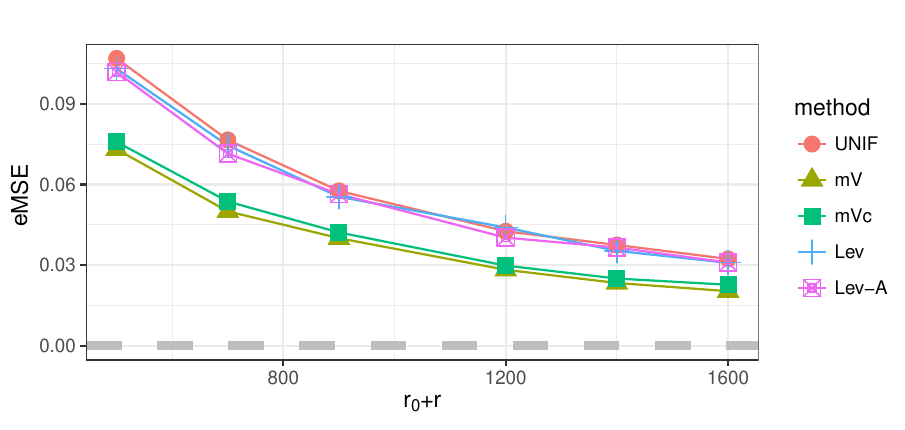}\\[-1cm]
    \caption{Case 4}
  \end{subfigure}\\[5mm]
  \caption{The eMSEs for NBR with different  subsample size ${r_0+r}$ and a fixed  $r_0/r=0.2$.
     The different distributions of covariates are listed in the beginning of Section 5.}
  \label{fig:fixratio-nb}
\end{figure}

From Figures \ref{fig:fixratio-poisson} and  \ref{fig:fixratio-nb},  we can see that our methods are slightly better than the cases that $r_0$ is fixed at 200. However, this improvement is not significant.

To explore influential factors on subsample sizes that have been discussed in Section \ref{sec:r-size} in terms of estimation accuracy, we consider additional four cases for the Poisson regression models listed as below. \begin{enumerate}[{Case S}1:]
\item The true value of $\bm\beta$ is a $7\times1$ vector of 0.5 and the covariates matrix {$X=\Sigma_n^{-1/2}\tilde{X}$. Here $\tilde{X}$ is the centralized version of  a $n\times 7$  matrix whose elements are i.i.d., generated from $U([-1,1])$, and $\Sigma_n$ is the sample covariance matrix of $\tilde{X}$ so that
  $X$ has a sample covariance matrix as $I_p$ and a condition number as 1.

\item The true value of $\bm\beta$ is a $14\times1$ vector whose first seven elements are set to be 0.5 and rest are set to be 0.1. The covariates matrix  $X=\Sigma^{-1/2}_n\tilde{X}$, where $\tilde{X}$ is the centralized version of a $n\times 14$ matrix whose elements are i.i.d. generated from $U([-1,1])$ and $\Sigma_n$ is the sample covariance matrix of $\tilde{X}$ so that the condition number of $X$ is 1 and the signal to noise ratio is nearly the same as that in Case S1.}

\item This case is the same as the Case S2 except that we replace $x_{i2}$ in Case S2 with $x_{i2}=x_{i1}+\varepsilon_{i}$ where $\varepsilon_{i}\overset{\text{i.i.d}}{\sim} U(\left[-0.4,0.4 \right] )$ for $i=1,\ldots,n$. For this setup, the condition number of $X$ is around 5.

\item This case is the same as the Case S2 except that we replace $x_{i2}$ in Case S2 with $x_{i2}=x_{i1}+\varepsilon_{i}$ where $\varepsilon_{i}\overset{\text{i.i.d}}{\sim} U(\left[-0.1,0.1 \right] )$ for $i=1,\ldots,n$. For this setup, the condition number of $X$ is around 26.

\end{enumerate}
To exclude the pilot subsampling effect,  the ideal case that $\hat{\bm\beta}_{\rm MLE}$ is given before conducting the subsampling
strategy is considered. Although this setting is hard to satisfy, the simulation provides some key insights for Theorem \ref{thm:r-size} and it is also valuable for the two step Algorithm.  The sample size $r$ is selected as 10, 15, 20, 25 and 30 times of the dimension  respectively.
For comparison, the uniform subsampling method  is also demonstrated. The empirical MSEs   are reported in Figures \ref{fig:q3_mse}.

\begin{figure}[!htp]
  \centering
  \begin{subfigure}{0.49\textwidth}
    \includegraphics[width=\textwidth]{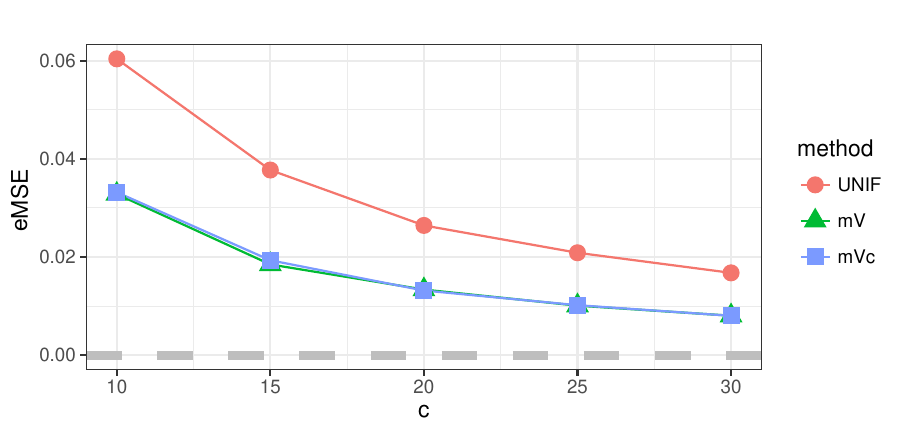}\\[-1cm]
    \caption{Case S1 ($p=7,\kappa=1$)}\label{fig:q34}
  \end{subfigure}
  \begin{subfigure}{0.49\textwidth}
    \includegraphics[width=\textwidth]{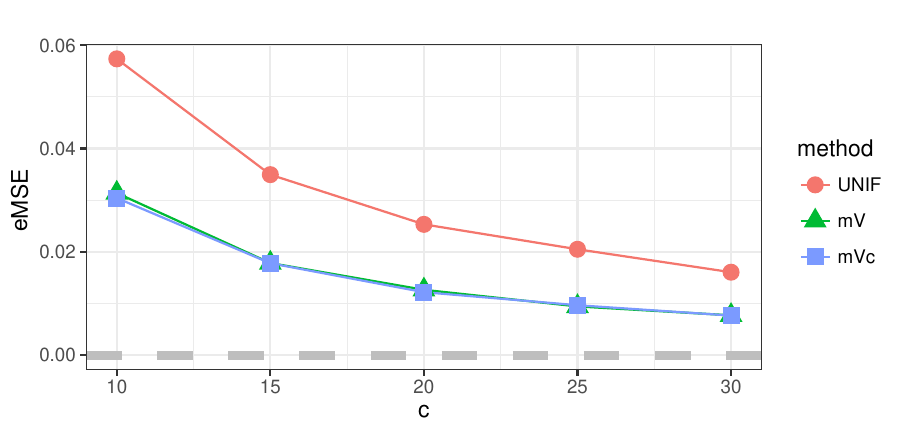}\\[-1cm]
    \caption{Case S2 ($p=14,\kappa=1$)}\label{fig:q31}
  \end{subfigure}
  \\[5mm]
  \begin{subfigure}{0.49\textwidth}
    \includegraphics[width=\textwidth]{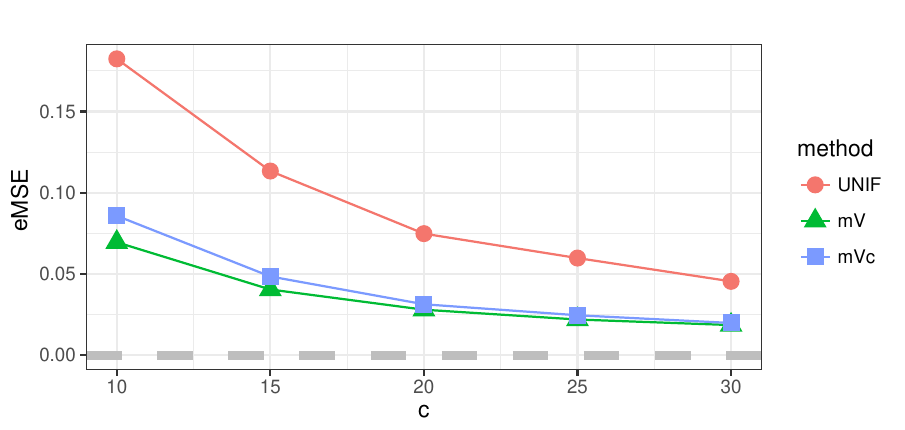}\\[-1cm]
    \caption{Case S3 ($p=14,\kappa\approx5$)}\label{fig:q33}
  \end{subfigure}
  \begin{subfigure}{0.49\textwidth}
    \includegraphics[width=\textwidth]{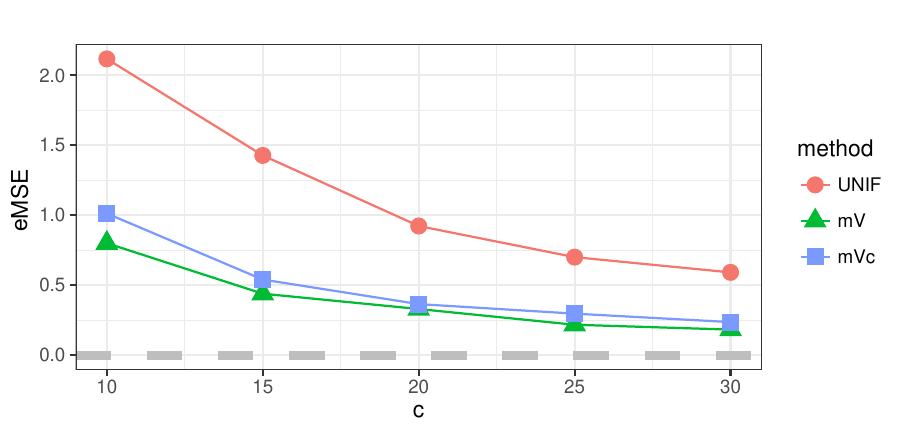}\\[-1cm]
    \caption{Case S4 ($p=14,\kappa\approx26$)}\label{fig:q32}
  \end{subfigure}
  \\[5mm]
  \caption{The eMSEs for Poisson regression with different  subsample size ${r=cp}$.
     The different distributions of covariates are listed in the beginning of Section S2.}
  \label{fig:q3_mse}
\end{figure}

 Through the simulation results reported in Figures S3(a)  and S3(b), we can see that the cases with $r=10p,\ 20p$ have similar performance when  the conditional numbers of the covariate matrix are fixed at one. And for the same dimensional case, the eMSEs become larger as the conditional number of the covariate matrix increasing. These echo the results  discussed in Section \ref{sec:r-size}.


\end{document}